\documentclass[10pt,a4paper]{article}
\makeatletter
\AtBeginDocument{%
  \@ifpackageloaded{hyperref}
  {\def\@doi#1{\href{https://doi.org/#1}
      {\ttfamily https://doi.org/#1}\egroup}}
  {\def\@doi#1{\ttfamily https://doi.org/#1\egroup}}
  \def\doi{\bgroup\catcode`\_=12\relax\@doi}}
\makeatother
\usepackage{authblk}

\usepackage{amssymb,amsmath,amsthm}
\usepackage{mathtools}
\usepackage[linesnumbered,ruled,vlined,noend]{algorithm2e}
\usepackage{graphicx}
\usepackage[latin1]{inputenc}
\usepackage[T1]{fontenc}
\usepackage{booktabs}
\usepackage{array}
\usepackage{units}
\usepackage{multirow,makecell}
\usepackage{pgfplots}
\usepackage{url}
\usepackage{float}
\usepackage{stmaryrd}
\usepackage{xfp}
\usepackage{bookmark}
\usepackage{hyperref}
\usepackage{siunitx}

\newtheorem{definition}{Definition}[section]
\newtheorem{theorem}{Theorem}[section]

\newtheorem{lemma}[theorem]{Lemma}

\hypersetup{
  colorlinks   = true, 
  urlcolor     = blue, 
  linkcolor    = blue, 
  citecolor   = red 
}

\usepackage{arydshln}

\renewcommand{\ge}{\geqslant}
\renewcommand{\geq}{\geqslant}
\newcommand{\ratioc}[2]{${#1}$ && ${#2}$ && $\times$\fpeval{trunc({#1} / {#2}, 2)}}

\newcommand{\fv}[1]{\ensuremath{\mathit{fv}({#1})}}
\newcommand{\succs}[1]{\ensuremath{\mathord{\downarrow}{#1}}}

\newcommand{\caesar}{\textsc{c{\ae}sar.bdd}}
\newcommand{\conc}{\ensuremath{\mathbin{\|}}}
\newcommand{\Conc}{\ensuremath{\mathrm{C}}}
\newcommand{\bconc}{\mathop{\mathcal{C}}}
\newcommand{\bind}{\mathop{\bar{\mathcal{C}}}}
\newcommand{\red}{\mathbin{{\rightarrow}\!{\bullet}}}
\newcommand{\agg}{\mathbin{{\circ}\!{\rightarrow}}}
\newcommand{\lvl}[1]{\mathrm{lvl}({#1})}
\newcommand{\maseq}[1]{\lfloor {#1} \rfloor}
\newcommand{\comma}{\mathbin{\raisebox{0.2ex}{\textbf{,}}}}
\newcommand{\Nat}{\mathbb{N}}

\newcommand{\reduc}{\vartriangleright}
\newcommand*{\defeq}{\triangleq}

\makeatletter
\def \rightarrowfill{\m@th\mathord{\smash-}\mkern-6mu%
  \cleaders\hbox{$\mkern-2mu\mathord{\smash-}\mkern-2mu$}\hfill
  \mkern-6mu\mathord\rightarrow}
\makeatother

\makeatletter
\def \Rightarrowfill{\m@th\mathord{\smash-}\mkern-6mu%
  \cleaders\hbox{$\mkern-2mu\mathord{\smash-}\mkern-2mu$}\hfill
  \mkern-6mu\mathord\Rightarrow}
\makeatother

\makeatletter
\def \rightarrowfill{\m@th\mathord{\smash-}\mkern-6mu%
  \cleaders\hbox{$\mkern-2mu\mathord{\smash-}\mkern-2mu$}\hfill
  \mkern-6mu\mathord\rightarrow}
\makeatother

\makeatletter
\def \Rightarrowfill{\m@th\mathord{\smash=}\mkern-6mu%
  \cleaders\hbox{$\mkern-2mu\mathord{\smash=}\mkern-2mu$}\hfill
  \mkern-6mu\mathord\Rightarrow}
\makeatother

\makeatletter
\def \midrightarrowfill{\m@th\mathord{\smash{\raisebox{.2ex}{$\scriptscriptstyle\mid$}}\!\!\,-}\mkern-6mu%
  \cleaders\hbox{$\mkern-2mu\mathord{\smash-}\mkern-2mu$}\hfill
  \mkern-6mu\mathord\rightarrow}
\makeatother

\makeatletter
\def \midRightarrowfill{\m@th\mathord{\smash{\raisebox{.1ex}{$\scriptstyle\mid$}}\!\!\!=}\mkern-6mu%
  \cleaders\hbox{$\mkern-2mu\mathord{\smash=}\mkern-2mu$}\hfill
  \mkern-6mu\mathord\Rightarrow}
\makeatother

\newcommand{\overstackrel}[2]{\mathrel{\mathop{#1}\limits^{#2}}}
\newcommand{\trans}[1]{\mathbin{\smash[t]{\overstackrel{\rightarrowfill}{\ #1\ }}}}

\newcommand{\tfg}[2][{}]{\llbracket {#2} \rrbracket_{#1}}

\usepackage{tikz}
\usetikzlibrary{arrows}
\usetikzlibrary{decorations.markings}
\usetikzlibrary{shadows}
\usepackage{subcaption}
\tikzset{
  big stealth/.style={
    decoration={markings,mark=at position -(0.1pt) with {\arrow[scale=2*\scale]{stealth}}},
    postaction={decorate},
    shorten >=0.4pt}}
\tikzset{
  big ring/.style={
    decoration={markings,mark=at position -(0.1pt) with {\arrow[scale=1.5*\scale]{o}}},
    postaction={decorate},
    shorten >=8pt*\scale}}
\tikzset{
  big disc/.style={
    decoration={markings,mark=at position -(0.1pt) with {\arrow[scale=1.5*\scale]{*}}},
    postaction={decorate},
    shorten >=8pt*\scale}}
\tikzset{
  big box/.style={
    decoration={markings,mark=at position -(0.1pt) with {\arrow[scale=1.5*\scale]{open square}}},
    postaction={decorate},
    shorten >=8pt*\scale}}
\tikzset{
  big tile/.style={
    decoration={markings,mark=at position -(0.1pt) with {\arrow[scale=1.5*\scale]{square}}},
    postaction={decorate},
    shorten >=8pt*\scale}}

\tikzstyle{state}=[circle, very thick, fill, top color=white, bottom color=white, draw=black, minimum size=40pt, drop shadow]
\tikzstyle{place}=[circle, very thick, fill, top color=white, bottom color=white, draw=black, minimum size=40pt, drop shadow]
\tikzstyle{bplace}=[circle, very thick, fill, top color=cyan, fill
opacity=0.2, text opacity=1, bottom color=white, draw=cyan, minimum size=40pt]
\tikzstyle{trans}=[rectangle, very thick, fill, top color=white, bottom color=white, draw=black, minimum size=32pt, drop shadow]
\tikzstyle{btrans}=[rectangle, very thick, fill, top color=cyan, fill
opacity=0.2, text opacity=1, bottom color=white, draw=cyan, minimum size=32pt]
\tikzstyle{arc}=[thick, big stealth, black]
\tikzstyle{barc}=[thick, big stealth, cyan]
\tikzstyle{read}=[thick, big disc, black]
\tikzstyle{inhibitor}=[thick, big ring, black]
\tikzstyle{stopwatch}=[thick, big tile, black]
\tikzstyle{stopwatchinhibitor}=[thick, big box, black]
\tikzstyle{priority}=[thick, big stealth, orange]
\tikzstyle{enabling}=[thick, big disc, orange]
\tikzstyle{disabling}=[thick, big ring, orange]
\tikzstyle{token}=[circle, fill, draw=black, minimum size=4pt]
\tikzstyle{glob-options}=[label
distance=6pt*\scalenodes*\scale,x=1pt,y=-1pt,scale=\scale,every
node/.style={transform shape}, baseline=(current bounding box.center)]
\tikzstyle{virtual}=[circle, draw=white, minimum size=20pt]

\usetikzlibrary{arrows.meta}
\tikzstyle{tnode}=[circle, thick, fill, top color=white, bottom color=white, draw=black, minimum size=30pt]
\tikzstyle{tred} = [thick, -{To}{Circle}]
\tikzstyle{tagg} = [thick, {Circle[open]}-{To}]

\title{Accelerating the Computation of Dead and Concurrent Places using Reductions}
\author[1]{Nicolas Amat}
\author[1]{Silvano Dal Zilio}
\author[1]{Didier {Le Botlan}}
\affil[1]{LAAS-CNRS, Universit\'{e} de Toulouse, CNRS, INSA, Toulouse, France}
\date{}
\begin{document}
\maketitle
\sloppy

\begin{abstract}
  We propose a new method for accelerating the computation of 
  a concurrency relation, that is all pairs of places in a Petri net
  that can be marked together. Our approach relies on a state space
  abstraction, that involves a mix between structural reductions and
  linear algebra, and a new data-structure that is specifically
  designed for our task.  Our algorithms are implemented in a tool,
  called Kong, that we test on a large collection of models used
  during the 2020 edition of the Model Checking Contest. Our
  experiments show that the approach works well, even when a moderate
  amount of reductions applies.
\end{abstract}

\section{Introduction}

We propose a new approach for computing the \emph{concurrency
  relation} of a Petri net, that is all pairs of places that can be
marked together in some reachable states. This problem has practical
applications, for instance because of its use for decomposing a Petri
net into the product of concurrent
processes~\cite{janicki_automatic_2020,garavel_nested-unit_2019}. It
also provides an interesting example of safety property that nicely
extends the notion of \emph{dead places}.
These problems raise difficult technical challenges and provide an
opportunity to test and improve new model checking
techniques~\cite{garavel2021proposal}.

Naturally, it is possible to compute the concurrency relation by first
computing the complete state space of a system and then checking,
individually, the reachability of each pair of places. But this
amounts to solving a quadratic number of reachability
properties---where the parameter is the number of places in the
net---and one would expect to find smarter solutions, even if it is
only for some specific cases. We are also interested in partial
solutions, where computing the whole state space is not feasible.

We recently became interested in this problem because we see it as a
good testbed for a new model checking technique that we are actively
developing~\cite{tacas,berthomieu2018petri,berthomieu_counting_2019}. It
is an abstraction technique, based on the use of structural
reductions~\cite{berthelot_transformations_1987}, that we successfully
implemented into a symbolic model checker called Tedd. The idea is to
compute reductions of the form $N_1 \reduc_E N_2$, where: $N_1$ is an
initial Petri net (that we want to analyse); $N_2$ is a residual net
(hopefully simpler than $N_1$); and $E$ is a system of linear
equations. The goal is to preserve enough information in $E$ so that
we can rebuild the reachable markings of $N_1$ knowing only those of
$N_2$. While there are many examples of the benefits of structural
reductions when model checking Petri nets, the use of an equation
system ($E$) for tracing back the effect of reductions is new, and we
are hopeful that this approach can be applied to other problems. For
example, we proved recently~\cite{tacas} that this approach also works
well when combined with SMT.

In this paper, we confirm that the same holds true when we tackle the
concurrent places problem. In practice, we can often reduce a net $N_1$
into another net $N_2$ with far fewer places. We show that we can
reconstruct the concurrency relation of $N_1$ from the one of $N_2$,
using a surprising and very efficient ``inverse transform'' that
depends only on $E$ and does not involve computing reachable
markings. (This is a model checking paper where no transitions are
fired!)  This is useful since the number of places is a predominant
parameter when computing the concurrency relation. Note that we are
not concerned with how to compute the relation on $N_2$, but only by
how we can \emph{accelerate} its calculation on $N_1$.

\paragraph*{Related Work.} Several works address the problem of
finding or characterizing the concurrent places of a Petri net. This
notion is mentioned under various names, such as \emph{coexistency
  defined by markings}~\cite{janicki_nets_1984}, \emph{concurrency
  graph}~\cite{wisniewski_prototyping_2018} or \emph{concurrency
  relation}~\cite{garavel_state_2004,kovalyov_concurrency_1992,kovalyov_polynomial_2000,semenov_combining_1995,wisniewski2019c}.
The main motivation is that the concurrency relation characterizes the
sub-parts, in a net, that can be simultaneously active. Therefore it
plays a useful role when decomposing a net into a collection
of independent components. This is the case in~\cite{wisniewski2019c},
where the authors draw a connection between concurrent places and the
presence of ``sequential modules (state machines)''. Another example
is the decomposition of nets into unit-safe NUPNs (Nested-Unit Petri
Nets)~\cite{janicki_automatic_2020,garavel_nested-unit_2019}, for which
the computation of the {concurrency relation} is one of the main
bottlenecks.

We know only a couple of tools that support the computation of the
concurrency relation. A recent tool is part of the Hippo
platform~\cite{wisniewski2019c}, available online. Our reference tool
is \caesar, from the CADP toolbox~\cite{pbhg2021,cadp}, that uses BDD
techniques to explore the state space of a net and find concurrent
places. It supports the computation of a partial relation and can
output the ``concurrency matrix'' of a net using a specific textual
format~\cite{garavel2021proposal}. We adopt the same format since we
use \caesar\ to compute the concurrency relation on the residual net,
$N_2$, and as a yardstick in our benchmarks.

Concerning our use of structural reductions, our main result can be
interpreted as an example of \emph{reduction
  theorem}~\cite{lipton_reduction_1975}, that allows to deduce
properties of an initial model ($N_1$) from properties of a simpler,
coarser-grained version ($N_2$). But our notion of reduction is more
complex and corresponds to the one pioneered by
Berthelot~\cite{berthelot_transformations_1987} (with the equations
added). Several tools use reductions for checking reachability
properties but none specializes in computing the concurrency
relation. We can mention TAPAAL~\cite{bonneland2019stubborn}, an
explicit-state model checker that combines partial-order reduction
techniques and structural reductions or, more recently, ITS
Tools~\cite{thierry-mieg_structural_2020}, which combines several
techniques, including structural reductions and the use of SAT and SMT
solvers.

\paragraph*{Outline and Contributions.}
We define the semantics of Petri nets and the notion of concurrent
places in Sect.~\ref{sec:petri-nets-polyh}.  This section also
introduces a simplified notion of ``reachability equivalence'', called
\emph{polyhedral abstraction}, that gives a formal definition to the
relation $N_1 \reduc_E N_2$.  Section~\ref{sec:tfg} contains our main
contributions. We describe a new data-structure, called Token Flow
Graph (TFG), that captures the particular structure of the equation
system generated with our approach. We prove several results on TFGs
that allow us to reason about the reachable places of a net by playing
a token game on this graph.  We use TFGs
(Sect.~\ref{sec:change_dimension_algorithm}) to define an algorithm
that implements our ``inverse transform'' and show how to adapt it to
situations where we only have partial knowledge of the residual
concurrency relation.  Our approach has been implemented and computing
experiments (Sect.~\ref{sec:experimental-results}) show that
reductions are effective on a large set of models. We perform our
experiments on an independently managed collection of Petri nets
($588$ instances) corresponding to the safe nets used during the 2020
edition of the Model Checking Contest~\cite{mcc2019}.  We observe
that, even with a moderate amount of reductions (say we can remove
$25\%$ 
of the places), we can compute complete results much faster with
reductions than without (often with speed-ups greater than $\times
100$). We also show that we perform well with incomplete relations,
where we are both faster and more accurate. We include the proofs of
all our results in the appendix.

\section{Petri Nets and Polyhedral Abstraction}
\label{sec:petri-nets-polyh}

A \textit{Petri net} $N$ is a tuple
$(P, T, \textbf{pre}, \textbf{post})$ where $P = \{p_1, \dots, p_n\}$
is a finite set of places, $T = \{t_1, \dots, t_k\}$ is a finite set
of transitions (disjoint from $P$), and
$\textbf{pre} : T \rightarrow (P \rightarrow \mathbb{N})$ and
$\textbf{post} : T \rightarrow (P \rightarrow \mathbb{N})$ are the
pre- and post-condition functions (also called the flow functions of
$N$). We often simply write that $p$ is a place of $N$ when $p \in P$.
A state $m$ of a net, also called a \emph{marking}, is a total mapping
$m : P \rightarrow \mathbb{N}$ which assigns a number of
\emph{tokens}, $m(p)$, to each place of $N$. A marked net $(N, m_0)$
is a pair composed of a net and its initial marking $m_0$.

A transition $t \in T$ is \textit{enabled} at marking $m \in \Nat^P$
when $m(p) \ge \textbf{pre}(t,p)$ for all places $p$ in $P$. (We can
also simply write $m \geq \textbf{pre}(t)$, where $\geq$ stands for
the component-wise comparison of markings.) A marking $m'$ is
reachable from a marking $m$ by firing transition $t$, denoted
$m \trans{t} m'$, if: (1) transition $t$ is enabled at $m$; and (2)
$m' = m - \textbf{pre}(t) + \textbf{post}(t)$. When the identity of
the transition is unimportant, we simply write this relation
$m \trans{} m'$.  More generally, marking $m'$ is reachable from $m$
in $N$, denoted $m \trans{}^\star m'$ if there is a (possibly empty)
sequence of reductions such that $m \trans{} \dots \trans{} m'$.  We
denote $R(N, m_0)$ the set of markings reachable from $m_0$ in $N$.

A marking $m$ is $k$-{bounded} when each place has at most $k$
tokens and a marked Petri net $(N, m_0)$ is bounded when there is a
constant $k$ such that all reachable markings are $k$-bounded. While
most of our results are valid in the general case---with nets that are
not necessarily bounded and without any restrictions on the flow
functions (the weights of the arcs)---our tool and our experiments
focus on the class of $1$-bounded nets, also called \emph{safe} nets.

\begin{figure}[htbp]
  \centering
  \includegraphics[width=0.85\textwidth]{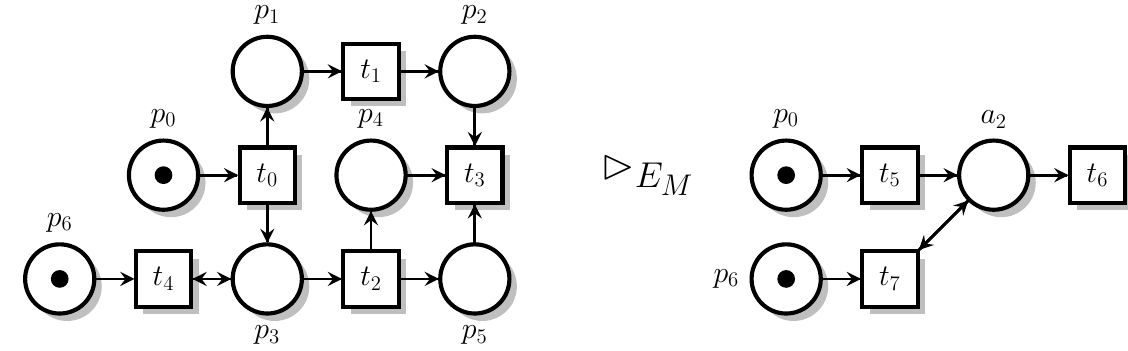}
  \caption{An example of Petri net, $M_1$ (left), and one of its polyhedral
    abstraction, $M_2$ (right), with $E_M \defeq (p_5 = p_4), (a_1 = p_1 + p_2),
    (a_2 = p_3 + p_4), (a_1 = a_2)$.\label{fig:stahl}}
\end{figure}

Given a marked net $(N, m_0)$, we say that places $p, q$ of $N$ are
concurrent when there exists a reachable marking $m$ with both $p$ and
$q$ marked.
The \textit{Concurrent Places} problem consists in enumerating all
such pairs of places.

\begin{definition}[Dead and Concurrent places]
  We say that a place $p$ of $(N, m_0)$ is \emph{not-dead}
  if there is $m$ in $R(N,m_0)$ such that
  $m(p) > 0$. In a similar way, we say that places $p,q$ are
  \emph{concurrent}, denoted $p \conc q$, if there is $m$ in
  $R(N,m_0)$ such that both $m(p) > 0$ and $m(q) > 0$. By extension,
  we use the notation $p \conc p$ when $p$ is not-dead. We say that
  $p, q$ are nonconcurrent, denoted $p \mathbin{\#} q$, when they are
  not concurrent.
\end{definition}

\paragraph*{Relation with Linear Arithmetic Constraints.} Many
results in Petri net theory are based on a relation with linear
algebra and linear programming
techniques~\cite{murata1989petri,silva1996linear}. A celebrated
example is that the potentially reachable markings of a net
$(N, m_0)$ are non-negative, integer solutions to the \emph{state
  equation} problem, $m = I \cdot \sigma + m_0$, with $I$ an integer
matrix defined from the flow functions of $N$ and $\sigma$ a vector in
$\Nat^k$.
It is known that solutions to the system of linear equations
$\sigma^T \cdot I = \vec{0}$ lead to \emph{place invariants},
$\sigma^T \cdot m = \sigma^T \cdot m_0$, that can provide some
information on the decomposition of a net into blocks of nonconcurrent
places, and therefore information on the concurrency relation.

For example, for net $M_1$ (Fig.~\ref{fig:stahl}), we can compute
invariant $p_4 - p_5 = 0$. This is enough to prove that places $p_4$
and $p_5$ are concurrent, if we can prove that at least one of them is
not-dead. Likewise, an invariant of the form $p + q = 1$ is enough to
prove that $p$ and $q$ are $1$-bounded and cannot be
concurrent. Unfortunately, invariants provide only an
over-approximation of the set of reachable markings, and it may be
difficult to find whether a net is part of the few known classes where the
set of reachable markings equals the set of potentially reachable
ones~\cite{hujsa:hal-02992521}.

Our approach shares some similarities with this kind of reasoning. A
main difference is that we will use equation systems to draw a
relation between the reachable markings of two nets; not to express
constraints about (potentially) reachable markings inside one
net. Like with invariants, this will allow us, in many cases, to
retrieve information about the concurrency relation without ``firing
any transition'', that is without exploring the state space.

In the following, we will often use place names as variables, and
markings $m : P \to \Nat$ as partial solutions to a set of linear
equations. For the sake of simplicity, all our equations will be of
the form $x = y_1 + \dots + y_l$ or $y_1 + \dots + y_l = k$ (with $k$
a constant in $\Nat$).

Given a system of linear equations $E$, we denote $\fv{E}$ the set of
all its variables. We are only interested in the non-negative integer
solutions of $E$. Hence, in our case, a \emph{solution} to $E$ is a
total mapping from variables in $\fv{E}$ to $\Nat$ such that all the
equations in $E$ are satisfied.  We say that $E$ is \emph{consistent}
when there is at least one such solution.
Given these definitions, we say that the mapping
$m: \{p_1, \dots, p_n\} \to \Nat$ is a (partial) solution of $E$ if
the system $E \comma \maseq{m}$ is consistent, with $\maseq{m}$ the
sequence of equations $p_1 = m(p_1) \comma \dots\comma p_n = m(p_n)$. (In some
sense, we use $\maseq{m}$ as a substitution.) For
instance, places $p, q$ are concurrent if the system
$p = 1 + x \comma q = 1 + y \comma \maseq{m}$ is consistent, where $m$ is a
reachable marking and $x, y$ are some fresh (slack) variables.

Given two markings $m_1 : P_1 \to \Nat$ and $m_2 : P_2 \to \Nat$, from
possibly different nets, we say that $m_1$ and $m_2$ are
\emph{compatible}, denoted $m_1 \equiv m_2$, if they have equal
marking on their shared places: $m_1(p) = m_2(p)$ for all $p$ in
$P_1 \cap P_2$. This is a necessary and sufficient condition for the
system $\maseq{m_1} \comma \maseq{m_2}$ to be consistent.

\paragraph*{Polyhedral Abstraction.}
We recently defined a notion of \emph{polyhedral abstraction} based on
our previous work applying structural reductions to model
counting~\cite{tacas,berthomieu_counting_2019}. We only need a
simplified version of this notion here, which entails an equivalence
between the state space of two nets, $(N_1, m_1)$ and $(N_2, m_2)$,
``up-to'' a system $E$ of linear equations.

\begin{definition}[$E$-equivalence]
  We say that $(N_1, m_1)$ is $E$-equivalent to\linebreak $(N_2, m_2)$, denoted
  $(N_1, m_1) \reduc_E (N_2, m_2)$, if and only if:
  \begin{description}
  \item [(A1)] $E \comma \maseq{m}$ is consistent for all markings $m$
    in $R(N_1, m_1)$ and $R(N_2, m_2)$;

  \item[(A2)] initial markings are \emph{compatible}, meaning
    $E \comma \maseq{m_1} \comma \maseq{m_2}$ is consistent;

  \item[(A3)] assume $m'_1, m'_2$ are markings of $N_1,N_2$,
    respectively, such that $E\comma \maseq{m'_1}\comma \maseq{m'_2}$ is
    consistent, then $m'_1$ is reachable if and only if $m'_2$ is
    reachable:\par $m'_1 \in R(N_1,m_1) \iff m_2' \in R(N_2, m_2)$.
  \end{description}
\end{definition}

By definition, relation $\reduc_E$ is symmetric. We deliberately use a
symbol oriented from left to right to stress the fact that $N_2$
should be a reduced version of $N_1$. In particular, we expect to have
less places in $N_2$ than in $N_1$.

Given a relation $(N_1, m_1) \reduc_E (N_2, m_2)$, each marking $m'_2$
reachable in $N_2$ can be associated to a unique subset of markings in
$N_1$, defined from the solutions to $E\comma \maseq{m'_2}$ (by
condition A1 and A3). We can show that this gives a partition of the
reachable markings of $(N_1, m_1)$ into ``convex sets''---hence the
name polyhedral abstraction---each associated to a reachable marking
in $N_2$. Our approach is particularly useful when the state space of
$N_2$ is very small compared to the one of $N_1$.  In the extreme
case, we can even find examples where $N_2$ is the ``empty'' net (a
net with zero places, and therefore a unique marking), but this
condition is not a requisite in our approach.

We can illustrate this result using the two marked nets $M_1, M_2$ in
Fig.~\ref{fig:stahl}, for which we can prove that
$M_1 \reduc_{E_M} M_2$. We have that
$m'_2 \defeq {a_2} = {1} \comma {p_6} = {1}$ is reachable in $M_2$,
which means that every solution to the system
$p_0 = 0 \comma p_1 + p_2 = 1 \comma p_3 + p_4 = 1 \comma p_4 = p_5
\comma p_6 = 1$ gives a reachable marking of $M_1$. Moreover, every
solution such that $p_i \geq 1$ and $p_j \geq 1$ gives a witness that
$p_i \conc p_j$. For instance, $p_1, p_4, p_5$ and $p_6$ are certainly
concurrent together.  We should exploit the fact that, under some
assumptions about $E$, we can find all such ``pairs of variables''
without the need to explicitly solve systems of the form
$E\comma \maseq{m}$; just by looking at the structure of $E$.

For this current work, we do not need to explain how to derive or check that an equivalence
statement is correct in order to describe our method. In practice, we
start from an initial net, $(N_1, m_1)$, and derive $(N_2, m_2)$ and $E$
using a combination of several structural reduction rules. You can
find a precise description of our set of rules
in~\cite{berthomieu_counting_2019} and a proof that the result of
these reductions always leads to a valid $E$-equivalence
in~\cite{tacas}. In most cases,
the system of linear equations obtained using this process exhibits a
graph-like structure. In the next section, we describe a set of
constraints that formalizes this observation. This is one of the
contributions of this paper, since we never defined something
equivalent in our previous works. We show with our benchmarks
(Sect.~\ref{sec:experimental-results}) that these constraints are
general enough to give good results on a large set of models.

\section{Token Flow Graphs}
\label{sec:tfg}

We introduce a set of structural constraints on the equations occurring
in an equivalence statement $(N_1, m_1) \reduc_E (N_2, m_2)$. The
goal is to define an algorithm that is able to easily compute
information on the concurrency relation of $N_1$, given the
concurrency relation on $N_2$, by taking advantage of the structure of
the equations in $E$.

We define the \textit{Token Flow Graph} (TFG) of a system $E$
of linear equations
as a Directed Acyclic Graph (DAG) with one vertex for each variable occurring
in $E$. Arcs in the TFG are used to depict the relation induced by equations in
$E$. We consider two kinds of arcs. Arcs for \emph{redundancy equations},
$q \red p$, to represent equations of the
form $p = q$ (or $p = q + r + \dots$), expressing that the marking of
place $p$ can be reconstructed from the marking of $q, r, \dots$
In this case, we say that place $p$ is
\emph{removed} by arc $q \red p$, because the marking of $q$ may influence the marking
of $p$, but not necessarily the other way round.

The second kind of arcs, $a
\agg p$, is for \emph{agglomeration equations}. It represents equations of the
form $a = p + q$, generated when we agglomerate several places into a new one.
In this case, we expect that if we can reach a marking with $k$
tokens in $a$, then we can certainly reach a marking with $k_1$ tokens in $p$
and $k_2$ tokens in $q$ when $k = k_1 + k_2$ (see property
Agglomeration in Lemma~\ref{lemma:forward_propagation}).  Hence information
flows in reverse order compared to the case of redundancy equations. This is
why, in this case, we say that places/nodes $p$ and $q$ are removed.
We also say that node $a$ is \emph{inserted}; it does not appear in $N_1$ but may
appear as a new place in $N_2$. We can have more than two places
in an agglomeration.

A TFG can also include nodes for
\emph{constants}, used to express invariant statements on the markings of the
form $p + q = k$. To this end, we assume that we have a family of disjoint sets
$K(n)$ (also disjoint from place and variable names), for each $n$ in $\Nat$,
such that the ``valuation'' of a node $v \in K(n)$ will always be $n$. We use
$K$ to denote the set of all constants.

\begin{definition}[Token Flow Graph]
  A TFG with set of places $P$ is a directed (bi)graph $(V, R, A)$
  such that: $V = P \cup S$ is a set of vertices (or nodes) with
  $S \subset K$ a finite set of constants; $R \in V\times V$ is a set
  of \emph{redundancy arcs}, $v \red v'$; and $A \in V\times V$ is a
  set of \emph{agglomeration arcs}, $v \agg v'$, disjoint from $R$.
\end{definition}

The main source of complexity in our approach arises from the need to
manage interdependencies between $A$ and $R$ nodes, that is situations
where redundancies and agglomerations alternate. This is not something
that can be easily achieved by looking only at the equations in $E$
and what motivates the need to define a specific data-structure.

We define several notations that will be useful in the
following.  We use the notation $v \rightarrow v'$ when we have
$(v \red v')$ in $R$ or $(v \agg v')$ in $A$.  We say that a node $v$
is a \textit{root} if it is never the target of an arc. A sequence of
nodes $(v_1, \dots, v_n)$ in $V^n$ is a \textit{path} if we have
$v_i \rightarrow v_{i+1}$ for all $i < n$. We use the notation
$v \rightarrow^\star v'$ when there is a path from $v$ to $v'$ in the
graph, or when $v = v'$.  We write $v \agg X$ when $X$ is the largest
subset $\{v_1, \dots, v_k\}$ of $V$ such that $X \neq \emptyset$ and
$v \agg v_i \in A$ for all $i \in 1..k$. Similarly, we write
$X \red v$ when $X$ is the largest, non-empty set of nodes
$\{v_1, \dots, v_k\}$ such that $v_i \red v \in R$ for all
$i \in 1..k$.

We display an example of Token Flow Graphs in
Fig.~\ref{fig:HuConstruction_TFG}, where ``black dot'' arcs
model edges in $R$ and ``white dot'' arcs model edges in $A$.
The idea is that each relation $X \red v$ or $v \agg X$
corresponds to one equation $v = \sum_{v_i \in X} v_i$ in $E$, and
that all the equations in $E$ should be reflected in the TFG.
We want to avoid situations where
the same place is removed more than once, or where some place occurs
in the TFG but is never mentioned in $N_1, N_2$ or $E$.  All these
constraints can be expressed using a suitable notion of well-formed
graph.

\begin{definition}[Well-Formed TFG]
  A TFG $G = (V, R, A)$ for the equivalence statement
  $(N_1, m_1) \vartriangleright_E (N_2, m_2)$ is \emph{well-formed}
  when all the following constraints are met, where $P_1$ and $P_2$ stand
  for the set of places in $N_1$ and $N_2$:
  \begin{description}
  \item[(T1)] \emph{no unused names:} $V \setminus K = P_1 \cup P_2 \cup \fv{E}$,
  \item[(T2)] \emph{nodes in $K$ are roots:} if $v \in V \cap K$ then
    $v$ is a root of $G$,
  \item[(T3)] \emph{nodes can be removed only once:} it is not possible to have
    $p \agg q$ and $p' \rightarrow q$ with $p \neq p'$, or to have
    both $p \red q$ and $p \agg q$,
  \item[(T4)] \emph{we have all and only the equations in $E$:} we
    have $v \agg X$ or $X \red v$ if and only if the equation
    $v = \sum_{v_i \in X} v_i$ is in $E$.
  \end{description}
\end{definition}

Given a relation $(N_1, m_1) \vartriangleright_E (N_2, m_2)$, the
well-formedness conditions are enough to ensure the unicity of a TFG (up-to the
choice of constant nodes) when we set each equation to be either in $A$ or in
$R$. In this case, we denote this TFG $\tfg{E}$. In practice, we use a tool
called Reduce to generate the $E$-equivalence from the initial net $(N_1, m_1)$.
This tool outputs a sequence of equations suitable to build a TFG and, for each
equation, it adds a tag indicating if it is a Redundancy or an
Agglomeration. We display in Fig.~\ref{fig:HuConstruction_TFG} the equations
generated by Reduce for the net $M_1$ given in Fig.~\ref{fig:stahl}.

\begin{figure}[tbp]
  \centering
  \begin{minipage}{0.2\linewidth}
    {\normalsize
\begin{verbatim}
# R |- p5 = p4
# A |- a1 = p2 + p1
# A |- a2 = p4 + p3
# R |- a1 = a2
\end{verbatim}}
    \end{minipage}\hspace{4em}
    \begin{minipage}[c]{0.45\linewidth}
      \def\scale{0.8} \def\scalenodes{1.0}
      \begin{tikzpicture}[glob-options]
        \node[tnode](p0) at (40.0, 50.0) {\large $p_0$};
        \node[tnode](p6) at (240.0, 50.0) {\large $p_6$};
        \node[tnode](a2) at (150.0, 50.0) {\large $a_2$};
        \node[tnode](a1) at (90.0, 80.0) {\large $a_1$};
        \node[tnode](p3) at (170.0, 100.0) {\large $p_3$};
        \node[tnode](p4) at (210.0, 80.0) {\large $p_4$};
        \node[tnode](p1) at (50.0, 130.0) {\large $p_1$};
        \node[tnode](p2) at (130.0, 130.0) {\large $p_2$};
        \node[tnode](p5) at (210.0, 130.0) {\large $p_5$};
        \draw[tred](a2) -- (a1); \draw[tred](p4) -- (p5);
        \draw[tagg](a2) -- (p3); \draw[tagg](a2) -- (p4);
        \draw[tagg](a1) -- (p1); \draw[tagg](a1) -- (p2);
      \end{tikzpicture}
    \end{minipage}          
    \caption{Equations generated from net $M_1$, in
      Fig.\ref{fig:stahl}, and associated TFG $\tfg{E_M}$}
    \label{fig:HuConstruction_TFG}
  \end{figure}
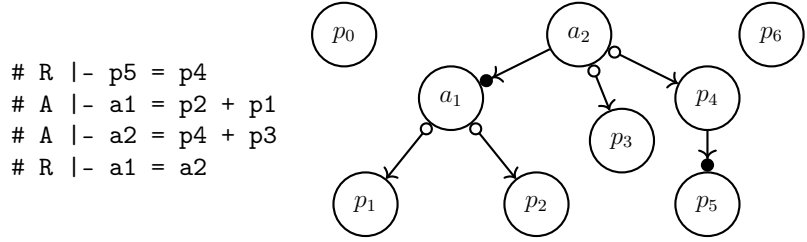
  
A consequence of condition (T3) is that a well-formed TFG is
necessarily {acyclic}; once a place has been removed, it cannot be
used to remove a place later. Moreover, in the case of reductions
generated from structural reductions, the {roots} of the graph are
exactly the constant nodes and the places that occur in $N_2$ (since
they are not removed by any equation). The constraints (T1)--(T4) are
not artificial or arbitrary. In practice, we compute $E$-equivalences
using multiple steps of structural reductions, and a TFG exactly
records the constraints and information generated during these
reductions. In some sense, equations $E$ abstract a relation between
the semantics of two nets, whereas a TFG records the structure of
reductions between places during reductions.

\paragraph*{Configurations of a Token Flow Graph.}
By construction, there is a strong connection between ``systems of
reduction equations'', $E$, and their associated graph, $\tfg{E}$. We
show that a similar relation exists between solutions of $E$ and
``valuations'' of the graph (what we call \emph{configurations}
thereafter).

A \textit{configuration} $c$ of a TFG $(V, R, A)$ is a partial
function from $V$ to $\Nat$. We use the notation $c(v)= \bot$ when $c$
is not defined on $v$ and we always assume that $c(v) = n$ when $v$ is a
constant node in $K(n)$.

Configuration $c$ is \textit{total} when $c(v)$ is defined for all
nodes $v$ in $V$; otherwise it is said \textit{partial}.  We use the
notation $c_{\mid N}$ for the configuration obtained from $c$ by
restricting its support to the set of places in the net $N$. We remark
that when $c$ is defined over all places of $N$ then $c_{\mid N}$ can
be viewed as a marking. By association with markings, we say that two
configurations $c$ and $c'$ are \textit{compatible}, denoted
$c \equiv c'$, if they have same value on the nodes where they are
both defined: $c(p) = c'(p)$ when $c(v) \neq \bot$ and
$c'(v) \neq \bot$.  We also use $\maseq{c}$ to represent the
system $v_1 = c(v_1) \comma \dots \comma v_k = c(v_k)$ where the $(v_i)_{i \in
  1..k}$ are the nodes such
that $c(v_i) \neq \bot$.
We say that a configuration $c$ is \emph{well-defined} when the
valuation of the nodes agrees with the equations associated with the
$A$ and $R$ arcs of $\tfg{E}$.
\begin{definition}[Well-Defined Configurations]
  Configuration $c$ is well-defined when for all nodes $p$ the
  following two conditions hold: \emph{\textbf{(CBot)}} if
  $v \rightarrow w$ then $c(v) = \bot$ if and only if $c(w) = \bot$;
  and \emph{\textbf{(CEq)}} if $c(v) \neq \bot$ and $v \agg X$ or
  $X \red v$ then $c(v) = \sum_{v_i \in X} c(v_i)$.
\end{definition}

We prove that the well-defined configurations of a TFG $\tfg{E}$ are
partial solutions of $E$, and reciprocally. Therefore, because all
the variables in $E$ are nodes in the TFG (condition T1) we have an
equivalence between solutions of $E$ and total, well-defined
configurations of $\tfg{E}$.

\begin{lemma}[Well-defined Configurations are
  Solutions]\label{lemma:configuration_satisfiability}
  Assume $\tfg{E}$ is a well-formed TFG for the equivalence
  $(N_1, m_1) \vartriangleright_E (N_2, m_2)$. If $c$ is a
  well-defined configuration of $\tfg{E}$ then $E\comma \maseq{c}$ is
  consistent. Conversely, if $c$ is a total configuration of $\tfg{E}$
  such that $E\comma \maseq{c}$ is consistent then $c$ is also
  well-defined.
\end{lemma}

We can prove several properties related to how the structure of a TFG
constrains possible values in well-formed configurations. These
results can be thought of as the equivalent of a ``token game'', which
explains how tokens can propagate along the arcs of a TFG. This is
useful in our context since we can assess that two nodes are
concurrent when we can mark them in the same configuration. (A similar
result holds for finding pairs of nonconcurrent nodes.)

Our first result shows that we can always propagate tokens from a node
to its children, meaning that if a node has a token, we can find one
in its \emph{successors} (possibly in a different well-defined
configuration). In the following, we use the notation $\succs{v}$ for
the set of successors of $v$, meaning:
$\succs{p} \defeq \bigcup\, \{ q \in V \,\mid\, p \rightarrow^\star q
\}$.
Property (Backward) states a dual result; if a child node is marked
then one of its parents must be marked.

\begin{lemma}[Token Propagation]\label{lemma:forward_propagation}
  Assume $\tfg{E}$ is a well-formed TFG for the equivalence
  $(N_1, m_1) \vartriangleright_E (N_2, m_2)$ and $c$ a well-defined
  configuration of $\tfg{E}$.\\[-1.5em]
  \begin{description}
  \item[(Forward)] if $p, q$ are nodes such that $c(p) \neq \bot$ and
  $p \rightarrow^\star q$ then we can find a well-defined
    configuration $c'$ such that $c'(q) \geqslant c'(p) = c(p)$ and
    $c'(v) = c(v)$ for every node $v$ not in $\succs{p}$.
  
  \item[(Backward)] if $c(p) > 0$ then there is a root $v$ such that
    $v \rightarrow^\star p$ and $c(v) > 0$.

  \item[(Agglomeration)] if $p \agg \{ q_1, \dots, q_k \}$ and
    $c(p) \neq \bot$ then for every sequence $(l_i)_{i \in 1..k}$ of
    $\Nat^k$, if $c(p) = \sum_{i \in 1..k} l_i$ then we can find a
    well-defined configuration $c'$ such that $c'(p) = c(p)$, and
    $c'(q_i) = l_i$ for all $i \in 1..k$, and $c'(v) = c(v)$ for every
    node $v$ not in $\succs{p}$.
  \end{description}
\end{lemma}

Until this point, none of our results rely on the properties of
$E$-equivalence. We now prove that there is an equivalence between
reachable markings and configurations of $\tfg{E}$. More precisely, we
prove (Th.~\ref{th:configuration_reachability}) that every reachable
marking in $N_1$ or $N_2$ can be extended into a well-defined
configuration of $\tfg{E}$. This entails that we can reconstruct all
the reachable markings of $N_1$ by looking at well-defined
configurations obtained from the reachable markings of $N_2$. Our
algorithm (see next section) will be a bit smarter since we do not
need to enumerate exhaustively all the markings of $N_2$. Instead, we
only need to know which roots can be marked together.

\begin{theorem}[Configuration Reachability]\label{th:configuration_reachability}
  Assume $\tfg{E}$ is a well-formed TFG for the equivalence
  $(N_1, m_1) \vartriangleright_E (N_2, m_2)$. If $m$ is a marking in
  $R(N_1, m_1)$ or $R(N_2, m_2)$ then there exists a total,
  well-defined configuration $c$ of $\tfg{E}$ such that $c \equiv
  m$. Conversely, given a total, well-defined configuration $c$ of
  $\tfg{E}$, if marking $c_{\mid N_1}$ is reachable in $(N_1, m_1)$
  then $c_{\mid N_2}$ is reachable in $(N_2, m_2)$.
\end{theorem}

\begin{proof}[Proof (sketch)]
  Take $m$ a marking in $R(N_1, m_1)$. By property of $E$-abstraction,
  there is a reachable marking $m_2'$ in $R(N_2, m_2)$ such that
  $E\comma \maseq{m}\comma \maseq{m_2'}$ is consistent. Therefore we
  can find a non-negative integer solution $c$ to the system
  $E\comma \maseq{m}\comma \maseq{m_2'}$. And $c$ is total because of
  condition (T1).  For the converse property, we assume that $c$ is a
  total and well-defined configuration of $\tfg{E}$ and that
  $c_{\mid N_1}$ is a marking of $R(N_1, m_1)$. By
  Lemma~\ref{lemma:configuration_satisfiability}, since $c$ is
  well-defined, we have that $E\comma \maseq{c}$ is consistent, and
  therefore so is
  $E\comma \maseq{c_{\mid N_1}}\comma \maseq{c_{\mid N_2}}$. This
  entails $c_{\mid N_2}$ in $R(N_2, m_2)$ by condition (A3), as
  needed.
\end{proof}

In the following, we will often consider that nets are safe. This is
not a problem in practice since our reduction rules preserve
safeness. Hence we do not need to check if $(N_2, m_2)$ is safe when
$(N_1, m_1)$ is.
The fact that the nets are safe has consequences. In particular, as a direct
corollary of Th.~\ref{th:configuration_reachability}, we can assume
that, for any well-defined configuration $c$, if $c_{\mid N_2}$ is reachable in
$(N_2, m_2)$ then $c(v) \in \{0, 1\}$.

By Th.~\ref{th:configuration_reachability}, if we take reachable
markings in $N_2$---meaning we fix the values of roots in
$\tfg{E}$---we can find places of $N_1$ that are marked together by
propagating tokens from the roots to the leaves
(Lemma~\ref{lemma:forward_propagation}). In our algorithm, next, we
show that we can compute the concurrency relation of $N_1$ by looking
at just two cases: (1) we start with a token in a single root $p$,
with $p$ not dead, and propagate this token forward until we find a
configuration with two places of $N_1$ marked together; or (2) we do
the same but placing a token in two separate roots, $p_1, p_2$, such
that $p_1 \conc p_2$.  We base our approach on the fact that we can
extend the notion of concurrent places (in a marked net), to the
notion of concurrent nodes in a TFG, meaning nodes that can be marked
together in a ``reachable configuration''.

\section{Dimensionality Reduction Algorithm}
\label{sec:change_dimension_algorithm}

We define an algorithm that takes as inputs a well-formed TFG
$\tfg{E}$ plus the concurrency relation for the net $(N_2, m_2)$, say
$\conc_2$, and outputs the concurrency relation for $(N_1, m_1)$, say
$\conc_1$. Actually, our algorithm computes a \emph{concurrency
  matrix}, $\Conc$, that is a matrix such that $\Conc[v,w] = 1$ when
the nodes $v, w$ can be marked together in a ``reachable
configuration'', and $0$ otherwise.
We prove
(Th.~\ref{th:matrix}) that the relation induced by $\Conc$ matches with
 $\conc_1$ on $N_1$. For the case of ``partial relations'', we use
$\Conc[v,w] = \bullet$ to mean that the relation is undecided. In this case
we say that matrix $\Conc$ is \textit{incomplete}.

The complexity of computing the concurrency relation is highly
dependent on the number of places in the net. For this reason, we say
that our algorithm performs some sort of a ``dimensionality
reduction'', because it allows us to solve a problem in a
high-dimension space (the number of places in $N_1$) by solving it
first on a lower dimension space (since $N_2$ may have far fewer
places) and then transporting back the result to the original net.
In practice, we compute the concurrency relation on $(N_2, m_2)$ using
the tool \caesar\ from the CADP toolbox; but we can rely on any kind
of ``oracle'' to compute this relation for us. This step is not
necessary when the initial net is fully reducible, in which case the
concurrency relation for $N_2$ is trivial and all the roots in
$\tfg{E}$ are constants.

We assume that $\tfg{E}$ is a well-formed TFG for the relation $(N_1, m_1)
\vartriangleright_E (N_2, m_2)$; that both nets are safe; and that all the roots
in $\tfg{E}$ are either constants (in $K(0) \cup K(1)$) or places in $N_2$. We
use symbol $\conc_2$ for the concurrency relation on $(N_2, m_2)$ and $\conc_1$
on $(N_1, m_1)$. To simplify our notations, we assume that $v \conc_2 w$ when
$v$ is a constant node in $K(1)$ and $w$ is not-dead. On the opposite, $v \mathbin{\#}_2 w$ when $v \in
K(0)$ or $w$ is dead.

Our algorithm is divided into two main functions, \ref{alg:change_of_dimension}
and~\ref{alg:token_propagation}. In the main function, \FuncSty{Matrix}, we
iterate over the non-dead roots of $\tfg{E}$ and recursively propagates a
``token'' to its successors (the call to \FuncSty{Propagate} in
line~\ref{line:alive}). After this step, we know all the live nodes
in $\Conc$. The call to \FuncSty{Propagate} has two effects. First, we retrieve
the list of successors of the live roots. Second, as a side-effect, we update
the concurrency matrix $\Conc$ by finding all the concurrent nodes that arise
from a unique root. We can prove all such cases arise from redundancy arcs that
are ``under node $v$''. Actually, we can prove that if $v \to w_1$  and $v \red
w_2$ (with $w_1 \neq w_2$) then the nodes in the set $\succs{v}
\setminus \succs{w_2}$ are concurrent to all the nodes in $\succs{w_2}$. Next,
in the second \KwSty{foreach} loop of \FuncSty{Matrix}, we compute the
concurrent nodes that arise from two distinct live roots $(v, w)$. In this case,
we can prove that all the successors of $v$ are concurrent with successors of
$w$: all the pairs in $\succs{v} \times \succs{w}$ are concurrent.

\begin{function}[t]
  \DontPrintSemicolon 

  \KwResult{the concurrency matrix $\Conc$}

  \BlankLine

  $\Conc \gets \vec{0}$ \tcc*{the matrix is initialized with zeros}
  
  \ForEach{root $v$ in $\tfg{E}$}
  {\If{$v \conc_2 v$}
    {$succs[v] \gets \FuncSty{Propagate}(\tfg{E}, \Conc, v)$}\label{line:alive}}
  
  \ForEach{pair of roots $(v,w)$ in $\tfg{E}$} {\If{$v \conc_2 w$}
    {\lForEach{$(v', w') \in \mathrm{succs}[v] \times \mathrm{succs}[w]$}
      {$\Conc[v', w'] \gets 1$}}}
  
  \KwRet{$\Conc$}

  \caption{Matrix($\tfg{E}$ : TFG, $\conc_2$ : \protect{concurrency relation on $(N_2, m_2)$})}
  \label{alg:change_of_dimension}
\end{function}

\begin{function}[t]
  \DontPrintSemicolon 

  \KwResult{the successors of $v$ in $\tfg{E}$. As a side-effect, we
    add to $\Conc$ all the relations that stem
    from knowing $v$ not-dead.}
  
  \BlankLine

  $\Conc[v,v] \gets 1$\label{line:non_dead}\;
  $\mathrm{succs} \gets \{v\}$ \tcc*{$\mathrm{succs}$ collects  the
    nodes in $\succs{v}$}
  $\mathrm{succr} \gets \{ \}$ \tcc*{auxiliary variable used to store
    $\succs{w}$ when $v \red w$}

  \lForEach{$w$ such that $v \agg w$}%
  {$\mathrm{succs} \gets \mathrm{succs} \cup
    \FuncSty{Propagate}(\tfg{E}, \Conc, w)$}

  \ForEach{$w$ such that $v \red w$}%
  {$\mathrm{succr} \gets \FuncSty{Propagate}(\tfg{E}, \Conc, w)$\;  
  \lForEach{$(v', w') \in \mathrm{succs} \times \mathrm{succr}$}
  {$\Conc[v',w'] \gets 1$}\label{line:redundancy_product}
  $\mathrm{succs} \gets \mathrm{succs} \cup \mathrm{succr}$}
    
  \Return $\mathrm{succs}$
  
  \caption{Propagate($\tfg{E}$ : TFG, $\Conc$ : concurrency matrix, $v$ : node)}
  \label{alg:token_propagation}
\end{function}

We can prove that our algorithm is sound and complete using the theory
that we developed on TFGs and configurations.
 
\begin{theorem}\label{th:matrix}
  If $\Conc$ is the matrix returned by a call to
  $\FuncSty{Matrix}(\tfg{E}, \conc_2)$ then for all places $p, q$ in
  $N_1$ we have $p \conc_1 q$ if and only if either $\Conc[p,q] = 1$
  or $\Conc[q,p] =1$.
\end{theorem}

We can perform a cursory analysis of the complexity of our
algorithm. By construction, we update the matrix by following the
edges of $\tfg{E}$, starting from the roots. Since a TFG is a DAG, it
means that we could call function \FuncSty{Propagate} several times on
the same node. However, a call to
$\FuncSty{Propagate}(\tfg{E}, \Conc, v)$ can only update $\Conc$ by
adding a $1$ between nodes that are successors of $v$ (information
only flows in the direction of $\to$). It means that
\FuncSty{Propagate} is idempotent; a subsequent call to
$\FuncSty{Propagate}(\tfg{E}, \Conc, v)$ will never change the values
in $\Conc$. As a consequence, we can safely memoize the result of this
call and we only need to go through a node at most once. More
precisely, we need to call \FuncSty{Propagate} only on the nodes that
are not-dead in $\tfg{E}$. During each call to \FuncSty{Propagate}, we
may update at most $O(N^2)$ values in $\Conc$, where $N$ is the number
of nodes in $\tfg{E}$, which is also $O(|\Conc|)$, the size of our
output. In conclusion, our algorithm has a linear time complexity (in
the number of live nodes) if we count the numbers of function calls
and a linear complexity, in the size of the output, if we count the
number of updates to $\Conc$. This has to be compared with the
complexity of building then checking the state space of the net, which
is PSPACE.

In practice, our algorithm is very efficient, highly parallelizable, and its
execution time is often negligible when compared to the other tasks involved
when computing the concurrency relation. We give some results on our
performances in the next section.

\paragraph*{Extensions to Incomplete Concurrency relations.}
With our approach, we only ever writes 1s into the concurrency matrix
$\Conc$. This is enough since we know relation $\conc_2$ exactly and,
in this case, relation $\conc_1$ must also be complete (we can have
only $0$s or $1$s in $\Conc$). This is made clear by the fact that
$\Conc$ is initialized with ${0}$s everywhere.
We can extend our algorithm to support the case where we only have a
partial knowledge of $\conc_2$. This is achieved by initializing
$\Conc$ with the special value $\bullet$ (undefined) and adding rules that
let us ``propagate $0$s'' on the TFG, in the same way that our total
algorithm only propagates $1$s. For example, we know that if
$\Conc[v,w]=0$ ($v,w$ are nonconcurrent) and $v \agg w'$ (we know that
always $c(v) \geq c(w')$ on reachable configurations) then certainly
$\Conc[w',w] =0$. Likewise, we can prove that following rule for
propagating ``dead nodes'' is sound: if $X \red v$ and $\Conc[w,w]=0$
(node $w$ is dead) for all $w \in X$ then $\Conc[v,v]=0$.

Partial knowledge on the concurrency relation can be useful. Indeed,
many use cases can deal with partial knowledge or only rely on the
nonconcurrency relation (a $0$ on the concurrency matrix). This is the
case, for instance, when computing NUPN partitions, where it is always
safe to replace a $\bullet$ with a $1$. It also means that knowing that
two places are nonconcurrent is often more valuable than knowing that
they are concurrent; $0$s are better than $1$s.

We have implemented an extension of our algorithm for the case of incomplete
matrices using this idea and we report some results obtained with it in the next
section. Unfortunately, we do not have enough space to describe the full
algorithm here. It is slightly more involved than for the complete case and is
based on a collection of six additional axioms:

\begin{itemize}
  \item If $\Conc[v,v] = 0$ then $\Conc[v, w] = 0$ for all node $w$ in
    $\tfg{E}$.
  \item If $v \agg X$ or $X \red v$ and $\Conc[w,w] = 0$ for all nodes $w \in X$
    then $\Conc[v,v] = 0$.
  \item If $v \agg X$ or $X \red v$ and $\Conc[v,v] = 0$ then $\Conc[w,w] = 0$ for
    all nodes $w \in X$.
  \item If $v \agg X$ or $X \red v$ then $\Conc[w,w'] = 0$ for all pairs
    of nodes $w, w' \in X$ such that $w \neq w'$.
  \item If $v \agg X$ or $X \red v$ and $\Conc[w,v'] = 0$ for all
    nodes $w \in X$ then $\Conc[v,v'] = 0$.
  \item If $v \agg X$ or $X \red v$ and $\Conc[v,v'] = 0$ then
    $\Conc[w,v'] = 0$ for all nodes $w$ in $X$.
\end{itemize}

While we can show that the algorithm is
sound, completeness takes a different meaning: we show that when
 nodes $p$ and $q$ are successors of roots $v_1$ and $v_2$ such
that $\Conc[v_i,v_i] \neq \bullet$ for all $i \in 1..2$ then
necessarily $\Conc[p,q] \neq \bullet$.

\section{Experimental Results}
\label{sec:experimental-results}

We have implemented our algorithm in a new tool, called {Kong} (for
{Koncurrent places Grinder}). The tool is open-source, under the GPLv3
license, and is freely available on GitHub
(\url{https://github.com/nicolasAmat/Kong}). We have used the
extensive database of models provided by the Model Checking Contest
(MCC)~\cite{mcc2019,HillahK17} to experiment with our approach. Kong
takes as inputs safe Petri nets defined using the Petri Net Markup
Language (PNML)~\cite{hillah2010pnml}. The tool does not compute net
reductions directly but relies on another tool, called
Reduce,
that is developed inside the Tina
toolbox~\cite{tina2004,tinaToolbox}. For our experiments, we also need
to compute the concurrency matrix of reduced nets.  This is done using
the tool \textsc{c{\ae}sar.bdd} (version 3.4, published in August
2020), that is part of the CADP toolbox~\cite{pbhg2021,cadp}, but we
could adopt any other technology here\footnote{we used version v3.4 of
  \textsc{c{\ae}sar.bdd}, part of CADP version 2020-h "Aalborg",
  published in August 2020.}.

\paragraph*{Benchmarks and Distribution of Reduction Ratios.}
Our benchmark is built from a collection of $588$ instances of safe
Petri nets used in the MCC 2020 competition. Since we rely on how much
reduction we can find in nets, we computed the {reduction ratio}
($r$), obtained using {Reduce}, on all the instances (see
Fig.~\ref{fig:reduction}). The ratio is calculated as the quotient
between how many places can be removed and the number of places in the
initial net. A ratio of $100\%$ ($r = 1$) means that the net is
\emph{fully reduced}; the residual net has no places and all the roots
are constants. We see that there is a surprisingly high number of
models whose size is more than halved with our approach (about $25$\%
of the instances have a ratio $r \ge 0.5$), with approximately half of
the instances that can be reduced by a ratio of $30\%$ or more. We
consider two values for the reduction ratio: one for reductions
leading to a well-formed TFG (in dark blue), the other for the best
possible reduction with Reduce (in light orange).

\begin{figure}[htbp]
	\centering
	\includegraphics[width=\linewidth]{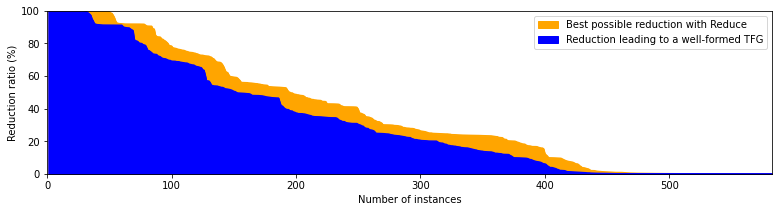}
	\caption{Distribution of reduction ratios over the safe instances in the
	MCC
	\label{fig:reduction}}  
\end{figure}

We observe that we lose few opportunities to reduce a net due to our
well-formedness constraint. Actually, we mostly lose the ability to
simplify some instances of ``partial'' marking graphs that could be
reduced using inhibitor arcs (a feature not supported by \caesar). We
evaluated the performance of Kong on the $424$ instances of safe Petri
nets with a reduction ratio greater than $1\%$. We ran Kong and
\textsc{c{\ae}sar.bdd} on each of those instances, in two main modes:
first with a time limit of $\SI{1}{\hour}$ to compare the number of
totally solved instances (when the tool compute a complete concurrency
matrix); next with a timeout of $\SI{60}{\second}$ to compare the
number of values (the filling ratios) computed in the partial
matrices.
Computation of a partial concurrency matrix with
\textsc{c{\ae}sar.bdd} is done in two phases: first a ``BDD
exploration'' phase that can be stopped by the user; then a
post-processing phase that cannot be stopped. In practice this means
that the execution time is often longer (because of the
post-processing phase) when we do not use Kong: the mean computation
time for \caesar\ alone is about $\SI{62}{\second}$, while it is less
than $\SI{21}{\second}$ when we use Kong and \caesar\ together.
In each test, we compared the output of Kong with the values obtained on the
initial net with \textsc{c{\ae}sar.bdd} and achieved $100\%$ reliability.

\paragraph*{Results for Totally  Computed Matrices.}
We report our results on the computation of complete matrices and a
timeout of $\SI{1}{\hour}$ in the table below.  We report the number
of computed matrices for three different categories of instances,
\emph{Low}/\emph{Fair}/\emph{High}, associated with different ratio
ranges. We observe that we can compute more results with reductions
than without ($+25\%$). As could be expected, the gain is greater on
category \emph{High} ($+53\%$), but it is still significant with the
\emph{Fair} instances ($+32\%$).

\newcolumntype{x}[1]{>{\centering\arraybackslash\hspace{0pt}}p{#1}}
\begin{center}
\resizebox{\textwidth}{!}{
  \begin {tabular}{l@{\quad}l c c c c x{1em} c c c}%
    \toprule
    \multicolumn{2}{c}{\multirow{2}{*}{\begin{minipage}[c]{6em}
          \centering\textsc{Reduction Ratio ($r$)} \end{minipage}}}
    && \multirow{2}{*}{
       \begin{minipage}[c]{6em}
         \centering \textsc{\# Test\\ Cases} \end{minipage}} &&
                                                                \multicolumn{3}{c}{\textsc{\# Computed
                                                                Matrices}}\\\cmidrule(rl){5-10} \multicolumn{2}{c}{} &&&&
                                                                                                                          \textsc{Kong} && {\textsc{\textsc{c{\ae}sar.bdd}}}\\\midrule
    \emph{Low} & $r \in \; ]0, 0.25[$ && $160$ && \ratioc{90}{88} \\
    \emph{Fair} & $r \in  [0.25, 0.5[$ && $112$ && \ratioc{53}{40} \\
    \emph{High} & $r \in  [0.5, 1]$ && $152$ && \ratioc{97}{63} \\\hline\\[-1em]
    {Total} & $r \in \; ]0, 1]$ && $\fpeval{160 + 112 + 152}$
                                                                                                                     && \ratioc{\fpeval{90 + 53 + 97}}{\fpeval{88 + 40 + 63}}\\
    \bottomrule
  \end{tabular}}
\end{center}

To understand the impact of reductions on the computation time, we
compare \textsc{c{\ae}sar.bdd} alone, on the initial net, and Kong +
Reduce + \textsc{c{\ae}sar.bdd} on the reduced net. We display the
result in a scatter plot, using a logarithmic scale
(Fig.~\ref{fig:partial_computations}, left), with one point for each
instance: time using reductions on the $y$-axis, and without on the
$x$-axis. We use colours to differentiate between \emph{Fair}
instances (light orange) and \emph{High} ones (dark blue), and fix a value of
$\SI{3600}{\second}$ when one of the computation timeout. Hence the
cluster of points on the right part of the plots are when \caesar\
alone timeouts.
We observe that the reduction ratio has a clear impact on the
speed-up and that almost all the data points are below the diagonal,
meaning reductions accelerate the computation in almost all cases,
with many test cases exhibiting speeds-up larger than $\times 10$ or
$\times 100$ (materialized by  dashed lines under the
diagonal).

\paragraph*{Results with Partial Matrices.}
We can also compare the ``accuracy'' of our approach when we have
incomplete results. To this end, we compute the concurrency relation
with a timeout of $\SI{60}{\second}$
on 
\textsc{c{\ae}sar.bdd}. We compare the \emph{filling ratio} obtained
with and without reductions. For a net with $n$ places, this ratio is
given by the formula ${2\, |\Conc|} / {(n^2 +n)}$, where $|\Conc|$ is
the number of $0$s and $1$s in the matrix. We display our results
using a scatter plot with linear scale, see
Fig.~\ref{fig:partial_computations} (right). Again, we observe that
almost all the data points are on one side of the diagonal, meaning in
this case that reductions increase the number of computed values, with
many examples (top line of the plot) where we can compute the complete
relation in $\SI{60}{\second}$ only using reductions. The graphic does
not discriminate between the number of $1$s and $0$s, but we obtain
similar good results when we consider the filling ratio for only the
concurrent places (the $1$s) or only the nonconcurrent places (the
$0$s).

\begin{figure}[tb]
  \centering
  {
    \begin{subfigure}[b]{.48\textwidth}
      \centering
      \includegraphics[width=\textwidth]{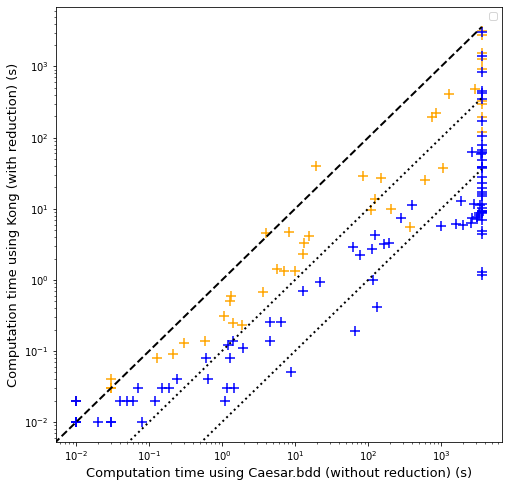}
    \end{subfigure}\hfill
    \begin{subfigure}[b]{.48\textwidth}
      \centering
      \includegraphics[width=\textwidth]{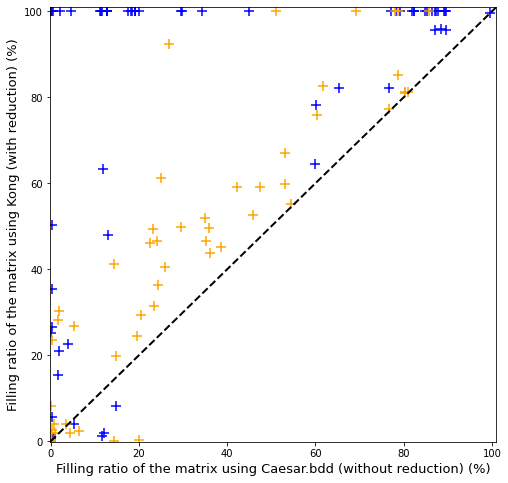}
    \end{subfigure}
  }
\caption{Comparing Kong ($y$-axis) and \textsc{c{\ae}sar.bdd}
  ($x$-axis) for instances with $r \in [0.25, 0.5[$ (light orange) and
  $r \in [0.5, 1]$ (dark blue). One diagram (left) compares the computation
  time for complete matrices; the other (right) compares the filling
  ratio for partial matrices with a timeout of $\SI{60}{\second}$.}
	\label{fig:partial_computations}
\end{figure}

\section{Conclusion and Further Work}
The concurrency problem is difficult, especially when we cannot
compute the complete state space of a net. We propose a method for
transporting this problem from an initial ``high-dimensionality''
domain (the set of places in the net) into a smaller one (the set of
places in the residual net). Our experiments confirm our intuition
that the concurrency relation is much easier to compute after
reductions (if the net can be reduced) and we provide an easy way to
map back the result into the original net.

Our approach is based on a combination of structural reductions with
linear equations first proposed
in~\cite{berthomieu2018petri,berthomieu_counting_2019}. Our main
contribution, in the current work, is the definition of a new
data-structure that precisely captures the structure of these linear
equations, what we call the Token Flow Graph (TFG). We use the TFGs to
accelerate the computation of the concurrency relation, both in the
complete and partial cases.
We have many ideas on how to apply TFGs to other problems and how to
extend them. A natural application would be for model counting (our
original goal in~\cite{berthomieu2018petri}), where the TFG could lead
to new algorithms for counting the number of (integer) solutions in
the systems of linear equations that we manage. Another possible
application is the \emph{max-marking} problem, which means finding the
maximum of the expression $\sum_{p \in P} m(p)$ over all reachable
markings. On safe nets, this amounts to finding the maximal number of
places that can be marked together.  We can easily adapt our algorithm
to compute this value and could even adapt it to compute the result
when the net is not safe.

We can even manage a more general problem, related to the notion of
\emph{max-concurrent} sets of places. We say that the set $S$ is
concurrent if there is a reachable $m$ such that $m(p) > 0$ for all
places $p$ in $S$. (This subsume the case of pairs and singleton of
places.) The set $S$ is
\emph{max-concurrent} if no superset $S'\supsetneq S$ is
concurrent. Computing the max-concurrent sets of a net is interesting
for several reasons. First, it gives an alternative representation of
the concurrency relation that can sometimes be more space efficient:
(1) the max-concurrent sets provide a unique cover of the set of
places of a net, and (2) we have $p \conc q$ if and only if there is
$S$ max-concurrent such that $\{p, q\} \subset S$. Obviously, on safe
nets, the size of the biggest max-concurrent set is the answer to the
\emph{max-marking} problem.

For future work, we would like to answer even more difficult
questions, such as proofs of Generalized Mutual Exclusion
Constraints~\cite{giua1992generalized}, that requires checking
invariants involving a weighted sums over the marking of places, of
the form $\sum_{p \in P} w_p .  m(p)$. Another possible extension
will be to support non-ordinary nets (which would require adding
weights on the arcs of the TFG) and nets that are not safe (which can
already be done with our current approach, but require changing some
of the ``axioms'' used in our algorithm).
Finally, another
interesting direction for works would be to find reductions that
preserve the concurrency relation (but not necessarily reachable
states). As you can see, there is a lot to be done, which underlines
the interest of studying TFGs.

\subsubsection*{Acknowledgements.} We would like to thank Pierre
Bouvier and Hubert Garavel for their insightful suggestions that
helped improve the quality of this paper.


\bibliographystyle{abbrv}
\bibliography{bibfile}

\newpage
\appendix

\section{Proofs}

\newtheorem*{maintheorem}{Theorem}
\newtheorem*{mainlemma}{Lemma}

\subsection{Proof of Lemma~\ref{lemma:configuration_satisfiability}: Well-defined
Configurations are  Solutions}
  
\begin{mainlemma}
  Assume $\tfg{E}$ is a well-formed TFG for the equivalence
  $(N_1, m_1) \reduc_E (N_2, m_2)$. If $c$ is a
  well-defined configuration of $\tfg{E}$ then $E\comma \maseq{c}$ is
  consistent. Conversely, if $c$ is a total configuration of $\tfg{E}$
  such that $E\comma \maseq{c}$ is consistent then $c$ is also
  well-defined.
\end{mainlemma}

\begin{proof}
  We prove each property separately.

  Assume $c$ is a well-defined configuration of $\tfg{E}$. Since $E$
  is a system of reduction equations, it is a sequence of equalities
  $\phi_1, \dots, \phi_k$ where each equation $\phi_i$ has the form
  $x_i = y_1 + \dots + y_{n}$. Also, since $\tfg{E}$ is well-formed we
  have that $X_i \red v_i$ or $x_i \agg X_i$ (only one case is
  possible) with $X_i = \{y_1, \dots, y_{n}\}$ for all indices
  $i \in 1..k$. We define $I$ the subset of indices in $1..k$ such
  that $c(x_i)$ is defined. By condition (CBot) we have
  $c(x_i) \neq \bot$ if and only if $c(v) \neq \bot$ for all
  $v \in X_i$. Therefore, if $c(x_i) \neq \bot$, we have by condition
  (CEq) that $\phi_i\comma \maseq{c}$ is consistent. Moreover the
  values of all the variables in $\phi_i$ are determined by
  $\maseq{c}$ (these variables have the same value in every
  solution). As a consequence, the system combining $\maseq{c}$
  and the $(\phi_i)_{i \in I}$ has a unique solution. On the opposite,
  if $c(x_i) = \bot$ then no variables in $\phi_i$ are defined by
  $\maseq{c}$. Nonetheless, we know that system $E$ is
  consistent. Indeed, by property of $E$-equivalence, we know that
  $E\comma \maseq{m_1}$ has solutions, so it is also the case with
  $E$. Therefore the system combining the equations in
  $(\phi_i)_{i \notin I}$ is consistent. Since this system shares no
  variables with the equations in $(\phi_i)_{i \in I}$, we have that
  $E\comma \maseq{c}$ is consistent.

  For the second case, we assume $c$ total and $E\comma \maseq{c}$
  consistent.  Since $c$ is total, condition (CBot) is true
  ($c(v) \neq \bot$ for all nodes in $\tfg{E}$). Assume we have
  $(N_1, m_1) \reduc_E (N_2, m_2)$. For condition (CEq), we
  rely on the fact that $\tfg{E}$ is well-formed. Indeed, for all
  equations in $E$ we have a corresponding relation $X \red v$ or
  $v \agg X$. Hence $E\comma \maseq{c}$ consistent implies that
  $c(v) = \sum_{w \in X} c(w)$.
\end{proof}


\subsection{Proof of Lemma~\ref{lemma:forward_propagation}: Token
  Propagation}

\begin{mainlemma}
  Assume $\tfg{E}$ is a well-formed TFG for the equivalence
  $(N_1, m_1) \reduc_E (N_2, m_2)$ and $c$ a well-defined
  configuration of $\tfg{E}$.
\end{mainlemma}
\begin{description}
\item[(Agglomeration)] if $p \agg \{ q_1, \dots, q_k \}$ and
  $c(p) \neq \bot$ then for every sequence $(l_1, \dots, l_k)$ in
  $\Nat^k$ such that $c(p) = \sum_{i \in 1..k} l_i$ we can find a
  well-defined configuration $c'$ such that $c'(p) = c(p)$, and
  $c'(q_i) = l_i$ for all $i \in 1..k$, and $c'(v) = c(v)$ for every
  node $v$ not in $\succs{p}$.
\item[(Forward)] for every pair $(p, q)$ of nodes such that $c(p) \neq \bot$
  and $p \rightarrow^\star q$ we can find a well-defined configuration $c'$
  such that $c'(q) \geqslant c'(p) = c(p)$ and $c'(v) = c(v)$ for every node
  $v$ not in $\succs{p}$.
\item[(Backward)] if $c(p) > 0$ then there is a root $v$ such that
  $v \rightarrow^\star p$ and $c(v) > 0$.
\end{description}

\begin{proof}
  We prove each property separately.\\[-1em]
  
  \noindent\textbf{(Agglomeration Propagation)}: we prove that we can update the
  successors of $p$ by following the order induced by the tree-like structure of
  the TFG. To this end, we introduce the notion of level of a node. The level of
  a node $v$ in $\tfg{E}$, denoted $\lvl{v}$, is the length of longest path $r
  \to^\star v$ from a root $r$ of $\tfg{E}$ to $v$. The level can only increase
  when we follow an arc and is always defined since we have no cycles in the
  TFG.

  Take a node $p$ at level $l$ such that $p \agg X$, with
  $X = \{ q_1, \dots, q_k \}$, and a sequence
  $(l_1, \dots, l_k) \in \Nat^k$ such that
  $c(p) = \sum_{i \in 1 ..  k} l_i$.
  We define configuration $c'$ as follows. Take $c'(q_i) = l_i$ for
  all $i \in 1..k$, and $c'(p) = c(p)$, and $c'(v) = c(v)$ for all the
  nodes $v$ such that $v \notin \succs{p}$ or $\lvl{v} \leqslant
  l$. We still need to find suitable values, $c'(w)$, for all the
  nodes $w$ that are successors of the nodes in $X$.  We proceed
  ``levels after levels''. Note that, by construction, we have that
  $\lvl{w} \geq l + 1$. If $c'(w)$ is in the last defined level and
  $w \agg Y$, then $c'$ cannot be already defined over $Y$ (otherwise
  it would mean that these nodes can be removed twice). In this case we
  are free to choose any possible valuation such that
  $\sum_{w'\in Y} c'(w') = c'(w)$ and we continue recursively with the
  successors of $Y$. If $Y \red w$ then all the nodes in $Y$ are in a
  level smaller than $w$ and therefore $c'$ is defined over $Y$. In
  this case we choose $c'(w) = \sum_{w'\in Y} c'(w')$. Since we have a
  finite DAG, this process terminates with $c'$ a total
  configuration. The proof proceeds by showing that $c'$ is a
  well-defined configuration, which is obvious.

  \noindent\textbf{(Forward Propagation)}: take a well-defined
  configuration $c$ of $\tfg{E}$ and assume we have two nodes $p,q$
  such that $c(p) \neq \bot$ and $p \rightarrow^\star q$.
  The proof is
  by induction on the length of the path from $p$ to $q$.
  The initial case is when $p = q$, which is trivial.
  Otherwise, assume $p \rightarrow r \rightarrow^\star q$. It is enough to find
  a well-defined configuration $c'$ such that $c'(r) \geqslant c'(p) = c(p)$.
  Since the nodes not in $\succs{p}$ are not in the paths from $p$ to $q$, we can
  ensure $c'(v) = c(v)$ for any node $v$ not in $\succs{p}$.
  The proof proceeds by a case analysis on $p \rightarrow r$.
  \begin{description}
  \item[(Case R)] assume $p \rightarrow r$ is a $R$-arc, meaning
    $p \red r$. More generally, it follows that $X \red r$ with
    $p \in X$.  Then by (CEq) we have
    $c(r) = c(p) + \sum_{v \in X, v \neq p} c(v) \geqslant c(p)$ and
    we can choose $c'= c$.
    
  \item[(Case A)] in this case we have $p \agg X$ with $r \in X$. By
    (Agglomeration Propagation) we can find a well-defined configuration $c'$
    such that $c'(r) = c'(p) = c(p)$ (and also $c'(v) = 0$ for all $v \in X
    \setminus \{ r\}$).
  \end{description}

  \noindent\textbf{(Backward Propagation)}: take a well-defined configuration
  $c$ of $\tfg{E}$ and assume we have $c(p) > 0$. The proof is by induction on
  the longest possible path from a root to node $p$. We re-use the notion of
  levels introduced previously. The initial case is when $p$ is itself a root,
  $\lvl{p} = 0$, and is trivial. Otherwise there must be at least one predecessor
  node $q$ such that $q \rightarrow p$. Like in the previous proof, we proceed
  by case analysis.
  \begin{description}
  \item[(Case R)] we have $X \red p$ with $X \neq \emptyset$.  By
    (CEq) we have $c(p) = \sum_{v \in X} c(v) > 0$. Hence there
    must be at least one node $q$ in $X$ such that $c(q) > 0$ and
    necessarily $\lvl{p} \geqslant \lvl{q} + 1$.
    
  \item[(Case A)] in this case we have $q \agg X$ with $p \in X$. By
    (CEq) we have
    $c(q) = c(p) + \sum_{v \in X\setminus \{p\}} c(v) \geqslant c(p)$
    as needed, with $\lvl{p} = \lvl{q} + 1$.
  \end{description}
\end{proof}


\subsection{Proof of Theorem~\ref{th:configuration_reachability}:
  Configuration Reachability}

\begin{maintheorem}
  Assume $\tfg{E}$ is a well-formed TFG for the equivalence
  $(N_1, m_1) \reduc_E (N_2, m_2)$. If $m$ is a marking in
  $R(N_1, m_1) \cup R(N_2, m_2)$ then there exists a total,
  well-defined configuration $c$ of $\tfg{E}$ such that $c \equiv
  m$. Conversely, if $c$ is a total, well-defined configuration of
  $\tfg{E}$ then marking $c_{\mid N_1}$ is reachable in $(N_1, m_1)$
  if and only if $c_{\mid N_2}$ is reachable in $(N_2, m_2)$.
\end{maintheorem}

\begin{proof}
  Take $m$ a marking in $R(N_1, m_1)$. The other case is totally
  symmetric. By property of $E$-abstraction, there exists a reachable
  marking $m_2'$ in $R(N_2, m_2)$ such that
  $E\comma \maseq{m}\comma \maseq{m_2'}$ is consistent. Therefore we can
  find a non-negative integer solution $c$ to the system
  $E\comma \maseq{m}\comma \maseq{m_2'}$, meaning a valuation for all
  the variables and places in $\fv{E}, N_1$ and $N_2$ such that
  $E\comma \maseq{c}$ is consistent and $c(p) = m(p)$ if $p \in N_1$ and
  $c(p) = m'_2(p)$ if $p \in N_2$. Because of condition (T1), this
  solution is total over all the nodes of $\tfg{E}$ (the only other
  possible case is for constants, whose values are fixed).

  For the converse property, we assume that $c$ is a total and
  well-defined configuration of $\tfg{E}$ and that $c_{\mid N_1}$ is a
  marking of $R(N_1, m_1)$.
  Since $c$ is a well-defined configuration, from
  Lemma~\ref{lemma:configuration_satisfiability} we have that
  $E\comma \maseq{c}$ is consistent. Therefore we have that
  $E\comma \maseq{c_{\mid N_1}}\comma \maseq{c_{\mid N_2}}$ is
  consistent. By definition of the $E$-abstraction, condition (A3), we
  have $c_{\mid N_2}$ in $R(N_2, m_2)$, as needed.
\end{proof}


\subsection{Safe Configurations}

For the sake of simplicity, we can assume that all the leaf nodes in
$\tfg{E}$ are places in $N_1$. This is true for TFGs computed from
structural reductions and this will simplify our proofs: for every
node $v$ we can always assume that there is $p$ in $N_1$ such that
$v \to^\star p$.

In our proof, we also implicitly assume that all the constants in $E$
are either $0$ or $1$. We could relax this last constraint, but this
would needlessly complicate our algorithm.

\begin{lemma}[Safe Configurations]\label{lemma:safe_configurations}
  Assume $\tfg{E}$ is a well-formed TFG for
  $(N_1, m_1) \reduc_E (N_2, m_2)$ with $(N_1, m_1)$ and
  $(N_2, m_2)$ safe Petri nets. Then for every total, well-defined
  configuration $c$ of $\tfg{E}$ such that $c_{\mid N_2}$ reachable in
  $(N_2,m_2)$, and every node $v$, we have $c(v) \in \{0, 1\}$.
\end{lemma}

\begin{proof}
  We prove the result by contradiction. Take a total and well-defined
  configuration $c$ such that $c_{\mid N_2}$ is reachable in
  $(N_2, m_2)$ and a node $v$ such that $c(v) > 1$ and
  $v \to^\star p$, with $p$ a place of $N_1$.
  By Lemma~\ref{lemma:forward_propagation}, we can find a well-defined
  configuration $c'$ of $\tfg{E}$ such that
  $c'(p) \geqslant c'(v) = c(v)$ and $c'(w) = c(w)$ for every node $w$
  not in $\succs{v}$.
  Therefore $c'_{\mid N_2}$ is also reachable in $(N_2, m_2)$.
  This contradicts the fact that the nets are safe since, by
  Th.~\ref*{th:configuration_reachability}, we would have a reachable
  marking that is not $1$-bounded.
\end{proof}


\subsection{Checking Dead  Places using Configurations}

By our configuration reachability theorem
(Th.~\ref{th:configuration_reachability}), if we take reachable
markings in $N_2$---meaning we fix the values of roots in
$\tfg{E}$---we can find places of $N_1$ that are marked together by
propagating tokens from the roots to the leaves
(Lemma~\ref{lemma:forward_propagation}). We prove that we can compute
the concurrency relation of $N_1$ by looking at just two cases: (1) we
start with a token in a single root $p$, with $p$ not dead, and
propagate this token forward until we find a configuration with two
places of $N_1$ marked together; or (2) we do the same but placing a
token in two separate roots, $p_1, p_2$, such that $p_1 \conc p_2$.
We base our approach on the fact that we can extend the notion of
concurrent places (in a marked net), to the notion of concurrent nodes
in a TFG. Those are the nodes that can be marked together in a
``reachable configuration''.

\begin{definition}[Concurrent Nodes]\label{def:bconc}
  The concurrency relation of $\tfg{E}$, denoted $\bconc$, is the
  relation between pairs of nodes in $\tfg{E}$ such that $v \bconc w$
  if and only if there is a total, well-defined configuration $c$
  where: (1) $c$ is reachable, meaning $c_{\mid N_2} \in R(N_2, m_2)$;
  and (2) $c(v) > 0$ and $c(w) > 0$.
\end{definition}

Like with the concurrency relation on nets, we have that $\bconc$ is
symmetric and $v \bconc v$ means that $v$ is not-dead (there is a
valuation with $c(v) > 0$). We can also extend this relation to define
a notion of max-concurrent sets of nodes.

By definition, if $p, q$ are places in $N_2$ then $p \bconc q$ only if
$p \conc q$ in $(N_2, m_2)$. We say in this case that $p, q$ are
\emph{concurrent roots}. We can extend this notion to constants. We
say that two roots $v_1, v_2$ are concurrent when $v_1 \bconc v_2$ and
that root $v_1$ is not-dead when $v_1 \bconc v_1$. This include cases
where $v_1$ or $v_2$ are in $K(1)$ (they are constants with value
$1$).

Since the places of $N_1$ are nodes in $\tfg{E}$, we also have that
$p \bconc q$ if and only if $p \conc q$ in $(N_1, m_1)$. This is the
relation we use in our algorithm of
Sect.~\ref{sec:change_dimension_algorithm}.

We prove some properties about the relation $\bconc$ that are direct
corollaries of our token propagation properties.
For all the following results, we implicitly assume that $\tfg{E}$ is a
well-formed TFG for the relation
$(N_1, m_1) \reduc_E (N_2, m_2)$, that both marked nets are
safe, that all the roots in $\tfg{E}$ are either constants or
places in $N_2$; and that $\bconc$ is the concurrency relation of $\tfg{E}$.

We start with a property (Lemma~\ref{lemma:non_dead_propagation})
stating that the successors of a ``live node'' must also be
not-dead. Lemma~\ref{lemma:non_dead_places_completeness} provides a dual
result, useful to prove the completeness of our approach; it states
that it is enough to explore the live roots to find all the live
nodes.

\begin{lemma}[Propagation of Live
  Nodes]\label{lemma:non_dead_propagation}
  If $v \bconc v$ and $v \to^\star w$ then $w \bconc w$.
\end{lemma}

\begin{proof}
  Assume $v \bconc v$. This means that there is a total, well-defined
  configuration $c$ such that $c(v) > 0$ and $c_{\mid N_2} \in R(N_2, m_2)$. Now
  take a successor node of $v$, say $v \to^\star w$. By
  Lemma~\ref{lemma:forward_propagation}, we can find another reachable
  configuration $c'$ such that $c'(w) \geqslant c'(v) = c(v)$ and $c'(x) =
  c(x)$ for all nodes $x$ not in $\succs{v}$. Therefore $w \bconc w$.
\end{proof}

\begin{lemma}[Live Nodes Come from Live
  Roots]\label{lemma:non_dead_places_completeness}
  If $v \bconc v$ then there is a root $v_0$ such that
  $v_0 \bconc v_0$ and $v_0 \to^\star {v}$.
\end{lemma}

\begin{proof}
  Assume $v \bconc v$. Then there is a total, well-defined
  configuration $c$ such that $c_{\mid N_2} \in R(N_2, m_2)$ and
  $c(v) > 0$. By the backward propagation property of
  Lemma~\ref{lemma:forward_propagation} we know that there is a root,
  say $v_0$, such that $c(v_0) \geq c(v)$ and $v_0 \to^\star v$. Hence
  $v_0$ is not-dead in $\tfg{E}$.
\end{proof}


\subsection{Checking Concurrent Places using Configurations}
\label{sec:check-conc-plac}

We can prove similar results for concurrent nodes instead of live ones. We
consider the two cases mentioned at the beginning of the section: when
concurrent nodes are obtained from two concurrent roots
(Lemma~\ref{lemma:concurrent_nodes_propagation}); or when they are
obtained from a single live root
(Lemma~\ref{lemma:concurrency_redundancy_nodes}), because of redundancy arcs.
Finally Lemma~\ref{lemma:concurrent_places_propagation} provides the associated
completeness result.

\begin{lemma}
  \label{lemma:concurrent_nodes_propagation} Assume $v,w$ are two
  nodes in $\tfg{E}$ such that $v \notin \succs{w}$ and
  $w \notin \succs{v}$. If $v \bconc w$ then $v' \bconc w'$ for all
  pairs of nodes $(v', w') \in \succs{v} \times \succs{w}$.
\end{lemma}
  
\begin{proof}
  Assume $v \bconc w$, $v \notin \succs{w}$ and $w \notin \succs{v}$.
  There must exist a total and well-defined configuration $c$ such
  that $c(v), c(w) > 0$ and $c_{\mid N_2} \in R(N_2, m_2)$.

  Take a successor $v'$ in $\succs{v}$, by applying the token
  propagation from Lemma~\ref{lemma:forward_propagation} we can
  construct a total and well-defined configuration $c'$ of $\tfg{E}$
  such that $c'(v') \geqslant c'(v) = c(v)$ and $c'(x) = c(x)$ for any
  node $x$ not in a $\succs{v}$. This is the case of $w$, hence $c'(w) = c(w) > 0$.

  We can use the token propagation property again, on $c'$. This gives
  a total and well-defined configuration $c''$ such that
  $c''(w') \geqslant c''(w) = c'(w) = c(w)$ and $c''(x) = c'(x)$ for
  any node $x$ not in $\succs{w}$.
  
  If we prove that $v' \notin \succs{w}$ we will then have
  $c''(v') = c'(v') \geqslant c(v)$, and therefore $v'\bconc w'$ as
  needed. We prove this result by contradiction. Indeed, assume
  $v' \in \succs{w}$. Hence,
  $\succs{v} \cap \succs{w} \neq \emptyset$.  Moreover, since $E$ is a
  well-formed TFG, there must exists (condition T3) three nodes
  $p,q,r$ such that $X \red r$, $p \in \succs{v} \cap X$ and
  $q \in \succs{w} \cap X$. In a similar way than the proof of
  Lemma~\ref{lemma:forward_propagation} we can propagate the tokens
  contained in $v,w$ to $p,q$, and obtain $c''(r) > 1$ from (CEq), which
  contradicts our assumption that the nets are safe.
\end{proof}

\begin{lemma}
  \label{lemma:concurrency_redundancy_nodes} If $v \bconc v$ and
  $v \red w$ then $v' \bconc w'$ for every pair of nodes $(v', w')$
  such that $v'\in (\succs(v) \setminus \succs{w})$ and
  $w' \in \succs{w}$.
\end{lemma}

\begin{proof}
  Assume $v \bconc v$ and $v \red w$. Hence there is a total,
  well-defined configuration $c$ such that $c(v) > 0$ and
  $c_{\mid N_2} \in R(N_2, m_2)$. Furthermore, since $v \red w$, we
  must have $c(w) > 0$ (condition CEq).

  Take $w'$ in $\succs{w}$. From Lemma~\ref{lemma:forward_propagation}
  we can find a total, well-defined configuration $c'$ such that
  $c'(w') \geqslant c'(w) = c(w) > 0$ and $c'(x) = c(x)$ for any node
  $x$ not in $\succs{w}$. Since $v$ is not in $\succs{w}$ we have
  $c'(v) = c(v)$. Likewise, places from $N_2$ are roots and therefore
  cannot be in $\succs{w}$. So we have
  $c'_{\mid N_2} \equiv c_{\mid N_2}$, which means $c'_{\mid N_2}$ is
  reachable in $(N_2, m_2)$. At this point we have $v \bconc w'$.

  Now, consider $v' \neq w$ such that $v \rightarrow v'$. We can use the forward
  propagation lemma a second time on $c'$ to find a total and well-defined
  configuration $c''$ such that $c''(v') \geqslant c''(v) = c'(v)$ and $c''(x) =
  c(x)$ for all nodes $x$ not in $\succs{v}$, and so, $c''_{\mid N_2}$
  is reachable in $(N_2, m_2)$. Since configuration $c''$ is well-defined we
  have (condition CEq) that $c''(v) = c''(w)$. 
  We also have $v' \notin \succs{w}$ and $w \notin \succs{v'}$ since $v
  \rightarrow w$, $v \rightarrow v'$ and $\tfg{E}$ is a well-formed TFG that
  must satisfy (T3). Finally, using
  Lemma~\ref{lemma:concurrent_nodes_propagation} is enough to prove that $v''
  \bconc w'$ for every node $v'' \in \succs{v'}$.
\end{proof}

\begin{lemma}
  \label{lemma:concurrent_places_propagation} If $v \bconc w$ and
  $v \neq w$ then one of the following two conditions is true.
  \begin{description}
  \item[{(Redundancy)}] There is a live node $v_0$ such that
    $v_0 \red w_0$ and either $(w, w)$ or $(w, v)$ are in
    $(\succs{v_0} \setminus \succs{w_0}) \times \succs{w_0}$.
  \item[{(Agglomeration)}] There is a pair of distinct roots
    $(v_0, w_0)$ such that $v_0 \bconc w_0$ with $v \in \succs{v_0}$
    and $w \in \succs{w_0}$.
  \end{description}
\end{lemma}

\begin{proof}
  Assume $v \bconc w$. Then there is a total, well-defined
  configuration $c$ such that $c_{\mid N_2} \in R(N_2, m_2)$
  and $c(u), c(v) = 1$ (the nets are safe).
  By the backward-propagation property in
  Lemma~\ref{lemma:forward_propagation} there exists two roots $v_0$
  and $w_0$ such that $c(v_0) = c(w_0) = 1$ with $v \in \succs{v_0}$
  and $w \in \succs{w_0}$. We need to consider two cases, either
  $v_0 \neq w_0$ or $v_0 = w_0$.

  The case where $v_0 \neq w_0$ corresponds to condition
  (Agglomeration).

  In the case where $v_0 = w_0$, we prove that there must be a node
  $v_1$ such that $v_0 \to^\star v_1$ and $v_1 \red w_1$ with either
  $(v, w)$ or $(w, v)$ in
  $(\succs{v_1} \setminus \succs{w_1}) \times \succs{w_1}$. We prove
  this result by contradiction. Indeed, if no such node exists then
  both $v$ and $w$ can be reached from $v_0$ by following only edges
  in $A$. Using the backward propagation property twice, and since
  $v \neq w$, this means that we can find a configuration $c'$ such
  that $c'_{\mid N_2} \equiv c_{\mid N_2}$ and
  $c'(v_0) \geq c(v) + c(w) \geq 2$, which contradicts our hypothesis
  that the nets are safe.
\end{proof}


\subsection{Proof of Theorem~\ref{th:matrix}: our Algorithm is Sound and Complete}

We prove a slightly different property that entails
Th.~\ref{th:matrix}. The following property makes use of the notations
introduced in the previous section and proves an equivalent result but
for all the nodes in $\tfg{E}$, not only for the places in $N_1$.

\begin{maintheorem}
  If $\Conc$ is the matrix returned by a call to
  $\FuncSty{Matrix}(\tfg{E}, \|)$, with $\|$ the concurrency relation
  between roots of $\tfg{E}$ (meaning $N_2$ augmented with the
  constants), then for all nodes $v, w$ we have $v \bconc w$ if and
  only if either $\Conc[v,w] = 1$ or $\Conc[w,v] =1$.
\end{maintheorem}

\begin{proof}
  We can first remark that the call to $\FuncSty{Matrix}(\tfg{E}, \|)$
  will always terminate.  We divide the proof into two different
  cases: first we prove that the computation of live nodes (the
  diagonal of $\Conc$ and the live nodes of $\bconc$) is sound and
  complete. Next, we prove the same result for concurrent nodes.

  \noindent\textbf{(Non-dead Places)} The result is a direct
  consequence of Lemmas~\ref{lemma:non_dead_propagation}
  and~\ref{lemma:non_dead_places_completeness}.

  \noindent\textbf{(Concurrent Places)} We need to consider the two
  cases describe in
  Lemma~\ref{lemma:concurrent_places_propagation}. The second
  \KwSty{foreach} loop in the code of \FuncSty{Matrix} takes care of
  the cases where concurrency is a consequence of two distinct live
  roots. The second case corresponds to the \KwSty{foreach} loop at
  line~\ref{line:redundancy_product} in the code of
  \FuncSty{Propagate}, for the matrix, and
  Lemma~\ref{lemma:concurrent_nodes_propagation} for
  $\bconc$. Finally,
  Lemma~\ref{lemma:concurrent_places_propagation} implies that this
  phase of the computation is 
  complete.
\end{proof}


\subsection{Axioms for Computing Incomplete Concurrency
  Matrices}\label{sec:axioms-comp-incompl}


Our algorithm for the case of incomplete matrices is based on a
collection of six additional axioms used to ``propagate $0$s'' in the
matrix $\Conc$.  We state each axiom separately and, in each case, we
prove a property that states that the axiom is sound. Completeness
takes a different meaning in this case. Indeed, we cannot prove that
we find all the pairs of nonconcurrent nodes. But we can prove a result
about the accuracy, meaning that all verdicts
$\Conc[v,w] \neq \bullet$ must originate from some roots $p, q$
such that $\Conc[p,q]$ is defined. More precisely, two concurrent
nodes ($\Conc[v,w] = 1$) must come from concurrent roots (or one live
root), and two nonconcurrent nodes ($\Conc[v,w] = 0$) imply that roots
leading to $v$ must all be nonconcurrent from roots leading to $w$.

In the following, we use the notation $v \bind w$ to say
$\neg (v \Conc w)$; meaning $v, w$ are nonconcurrent according to
$\Conc$. With our notations, $v \bind v$ means that $v$ is dead: there
is no well-defined, reachable configuration $c$ with $c(v) > 0$.

\subsubsection{Propagation of Dead Nodes.}
We prove that a dead node, $v$, is necessarily nonconcurrent from all
the other nodes. Also, if all the ``direct successors'' of a node are
dead then also is the node.

\begin{lemma}
  \label{lemma:independent_nodes_propagation_1} Assume $v$ a node in
  $\tfg{E}$.  If $v \bind v$ then for all nodes $w$ in $\tfg{E}$ we
  have $v \bind w$.
\end{lemma}

\begin{proof}
  Assume $v \bind v$. Then for any total and well-defined
  configuration $c$ such that $c_{\mid N_2}$ is reachable in
  $(N_2, m_2)$ we have $c(v) = 0$. By definition of the concurrency
  relation $\bconc$, $v$ cannot be concurrent to any node.
\end{proof}

\begin{lemma}\label{lemma:dead_node} Assume $v$ a node in $\tfg{E}$ such that $v
  \agg X$ or $X \red v$. Then $v \bind v$ if and only if $w \bind w$ for all
  nodes $w$ in $X$.
\end{lemma}

\begin{proof}
  We prove by contradiction both directions.
  
  Assume $v \bind v$ and take $w \in X$ such that $w \bconc w$. Then there
  is a total, well-defined configuration $c$ such that $c(w) > 0$.
  Necessarily, since $v \bind v$ we have $c(v) = 0$, which contradicts (CEq).

  Next, assume $v \bconc v$ and $w \bind w$ for every node $w \in X$. Then there
  is a total, well-defined configuration $c$ such that $c(v) > 0$.
  Necessarily, for all nodes $w \in X$ we have $c(w) = 0$, which also contradicts
  (CEq).
\end{proof}

These properties imply the soundness of the following three axioms:
\begin{itemize}
\item If $\Conc[v,v] = 0$ then $\Conc[v, w] = 0$ for all node $w$ in
  $\tfg{E}$.
\item If $v \agg X$ or $X \red v$ and $\Conc[w,w] = 0$ for all nodes $w \in X$
  then $\Conc[v,v] = 0$.
\item If $v \agg X$ or $X \red v$ and $\Conc[v,v] = 0$ then $\Conc[w,w] = 0$ for
  all nodes $w \in X$.
\end{itemize}

\subsubsection{Independency between Siblings.}
We prove that direct successors of a node are nonconcurrent from each
other (in the case of safe nets). This is basically a consequence of
the fact that $c(v) = c(w) + c(w') + \dots$ and $c(v) \leqslant 1$
implies that at most one of $c(w)$ and $c(w')$ can be equal to $1$
when the configuration is fixed.

Like with our ``safeness property'', we assume for the sake of
simplicity that all the leaves in $\tfg{E}$ are places in $N_1$.

\begin{lemma}
  \label{lemma:independent_nodes_propagation_2} Assume $v$ a node in
  $\tfg{E}$ such that $v \agg X$ or $X \red v$. For every pair of
  nodes $w,w'$ in $X$, we have that $w \neq w'$ implies
  $w \bar{C} w'$.
\end{lemma}

\begin{proof}
  The proof is by contradiction.  Take a pair of distinct nodes $w,w'$
  in $X$ and assume $w \bconc w'$. Then there exists a total and
  well-defined configuration $c$ such that $c(w) \geqslant 1$ and
  $ c(w') \geqslant 1$, with $c_{\mid N_2}$ reachable in $(N_2,
  m_2)$. Since $c$ must satisfy (CEq) we have $c(v) \geqslant 2$,
  which contradicts the fact that our nets are safe, see
  Lemma~\ref{lemma:safe_configurations}.
\end{proof}

This property implies the soundness of the following axiom:
\begin{itemize}
\item If $v \agg X$ or $X \red v$ then $\Conc[w,w'] = 0$ for all pairs
  of nodes $w, w' \in X$ such that $w \neq w'$.
\end{itemize}

\subsubsection{Heredity and Independency.}
We prove that if $v$ and $v'$ are nonconcurrent, then $v'$ must be
nonconcurrent from all the direct successors of $v$ (and
reciprocally). This is basically a consequence of the fact that
$c(v) = c(w) + \dots$ and $c(v) + c(v') \leqslant 1$ implies that
$c(w) + c(v') \leqslant 1$.

\begin{lemma}
  \label{lemma:independent_nodes_propagation_3} Assume $v$ a node in
  $\tfg{E}$ such that $v \agg X$ or $X \red v$. Then for every node
  $v'$ such that $v \bind v'$ we also have $w \bind v'$ for every node
  $w$ in $X$. Conversely, if $w \bind v'$ for every node $w$ in $X$
  then $v \bind v'$.
\end{lemma}

\begin{proof}
  We prove by contradiction each property separately.

  Assume $v \bind v'$ and take $w \in X$ such that $w \bconc v'$. Then there
  is a total, well-defined configuration $c$ such that $c(w),c(v') > 0$.
  Necessarily, since $v \bind v'$ we must have $c(v) = 0$ or $c(v') = 0$. We
  already know that $c(v') > 0$, so $c(v) = 0$, which contradicts (CEq) since $w
  \in X$.
  
  Next, assume $w \bind v'$ for all nodes $w \in X$ and we have $v \bconc v'$.
  Then there is a total, well-defined configuration $c$ such that
  $c(v),c(v') > 0$. Necessarily, for all nodes $w \in X$ we have $c(w) = 0$ or
  $c(v') = 0$. We already know that $c(v') > 0$, so $c(w) = 0$ for all nodes $w
  \in X$, which also contradicts (CEq).
\end{proof}

These properties imply the soundness of the following two axioms:
\begin{itemize}
\item If $v \agg X$ or $X \red v$ and $\Conc[w,v'] = 0$ for all
  nodes $w \in X$ then $\Conc[v,v'] = 0$.
\item If $v \agg X$ or $X \red v$ and $\Conc[v,v'] = 0$ then
  $\Conc[w,v'] = 0$ for all nodes $w$ in $X$.
\end{itemize}





\end{document}